\documentclass[10pt]{article}
\usepackage{times}
\usepackage{fullpage}
\usepackage{amsmath}
\usepackage{amsthm}
\usepackage{amsfonts}
\usepackage{amssymb}
\usepackage{latexsym}
\usepackage{epsfig,graphicx}
\usepackage{amsmath, amsthm, amssymb}
\usepackage{amsfonts}
\usepackage{graphics}
\usepackage{psfrag}
\usepackage{graphicx}
\usepackage{color}

\definecolor{Red}{rgb}{1,0,0}
\definecolor{Blue}{rgb}{0,0,1}
\definecolor{Olive}{rgb}{0.41,0.55,0.13}
\definecolor{Green}{rgb}{0,1,0}
\definecolor{MGreen}{rgb}{0,0.8,0}
\definecolor{DGreen}{rgb}{0,0.55,0}
\definecolor{Yellow}{rgb}{1,1,0}
\definecolor{Cyan}{rgb}{0,1,1}
\definecolor{Magenta}{rgb}{1,0,1}
\definecolor{Orange}{rgb}{1,.5,0}
\definecolor{Violet}{rgb}{.5,0,.5}
\definecolor{Purple}{rgb}{.75,0,.25}
\definecolor{Brown}{rgb}{.75,.5,.25}
\definecolor{Grey}{rgb}{.5,.5,.5}
\definecolor{Black}{rgb}{0,0,0}

\newtheorem{thm}{Theorem}
\newtheorem{Algo}[thm]{Algorithm}
\newtheorem{cor}[thm]{Corollary}

\newtheorem{lem}[thm]{Lemma}

\newtheorem{proposition}[thm]{Proposition}

\newtheorem{remark}[thm]{Remark}
\numberwithin{equation}{section}

\newcommand{\BP}{\ensuremath{\operatorname{BP}}}

\newcommand{\vecone}{\vec{1}}

\newcommand{\tensor}{\otimes}

\newcommand{\id}{\bc{\begin{array}{ccc}1&0&0\\0&1&0\\0&0&1\end{array}}}

\long\def\comment#1{}

\makeatletter
\def\@cite#1#2{[\if@tempswa #2 \fi #1]}
\makeatother

\newcommand{\eps}{\varepsilon}

\newcommand{\Consistent}[1]{\ensuremath{A}}

\newcommand{\Reachable}[1]{\ensuremath{B}}

\newcommand{\ra}{\rightarrow}

\newcommand\pr{\mathrm{P}}

\newcommand{\bink}[2]
    {{{#1}\choose {#2}}}

\newcommand\bc[1]{\left({#1}\right)}
\newcommand\cbc[1]{\left\{{#1}\right\}}

\newcommand{\bck}[1]{\left\langle{#1}\right\rangle}
\newcommand\brk[1]{\left\lbrack{#1}\right\rbrack}
\newcommand\scal[2]{\bck{{#1},{#2}}}

\newcommand\RR{\mathbf{R}}

\newcommand\algstyle{\sffamily\small}

\newcommand{\LL}{\mathcal{L}}
\newcommand{\MM}{\mathcal{M}}
\newcommand{\BB}{\mathcal{B}}
\newcommand{\KK}{\mathcal{K}}
\newcommand{\AAA}{\mathcal{A}}
\newcommand{\RRR}{\mathcal{R}}

\newcommand{\EE}{\mathcal{E}}

\def\eps{\epsilon}

\def\N{\hbox{I\kern-.2em\hbox{N}}}
\def\R{\hbox{I\kern-.2em\hbox{R}}}

\def\L{{\cal{L}}}

\def\0{\hat{0}}
\def\1{\hat{1}}

\def\|{\, || \, }

\newcommand{\anote}[1]{}
\newcommand{\enote}[1]{}
\newcommand{\mnote}[1]{}

\newsavebox{\s}
\newsavebox{\sz}
\newsavebox{\speicher}
\newlength{\laenge}
\newlength{\laengez}
\sbox{\s}{\(\cup\)} \sbox{\sz}{\(\cdot\)}
\sbox{\speicher}{\usebox{\s}\hspace*{-.5\laenge}\usebox{\sz}\hspace*{.5\laenge}\hspace*{-.5\laengez}}
\newcommand{\du}{\usebox{\speicher}}

\title{ {A Spectral Approach to Analyzing Belief Propagation for 3-Coloring} \ }
\author{ Amin Coja-Oghlan\thanks{Carnegie Mellon University, Department of Mathematical Sciences. Email: {\tt amincoja@andrew.cmu.edu}. Supported by DFG CO~646.} \and
Elchanan Mossel\thanks{U.C. Berkeley. E-mail: {\tt
        mossel@stat.berkeley.edu}. Supported by a Sloan fellowship
  in Mathematics, by NSF Career award DMS-0548249, NSF grant
  DMS-0528488 and ONR grant N0014-07-1-05-06}
\and Dan Vilenchik\thanks{Tel-Aviv University. E-mail: {\tt
        vilenchi@post.tau.ac.il}.}
        }

\begin{document}
\maketitle
\begin{abstract}
Belief Propagation ($\BP$) is a message-passing
algorithm that computes the exact marginal
distributions at every vertex of a graphical model without cycles.
While $\BP$ is designed to work correctly on trees, it is routinely
applied to general graphical models that may contain cycles, in
which case neither convergence, nor correctness in the case of
convergence is guaranteed.
Nonetheless, $\BP$ gained popularity as it seems to remain effective in many cases of
interest, even when the underlying graph is ``far" from being a tree.
However, the theoretical understanding of $\BP$ (and its new relative Survey
Propagation) when applied to CSPs is poor.

Contributing to the rigorous understanding of BP,
in this paper we relate the convergence of BP
to spectral properties of the graph.
This encompasses a result for random graphs with a ``planted'' solution;
thus, we obtain the first rigorous result on BP for graph coloring in the
case of a complex graphical structure (as opposed to trees).
In particular, the analysis shows
how Belief Propagation breaks the symmetry between the $3!$ possible
permutations of the color classes.

\end{abstract}
\bigskip



\section{Introduction and Results}

\subsection{Message Passing Algorithms}

This paper deals with a rigorous analysis of the Belief Propagation (``BP'' for short) algorithm on certain instances
of the 3-coloring problem.
Originally BP was introduced by Pearl~\cite{Pearl} as a message passing algorithm to compute the marginals at the vertices
of a probability distribution described by an acyclic ``graphical model'', i.e., a
representation of the distribution's dependency structure as an acyclic graph.
Although in the worst case BP will fail if the graphical representation features cycles,
various version of BP are in common use as heuristics in artificial intelligence
and statistics, where they frequently perform well empirically
as long as the underlying model does at least not contain (many) ``short'' cycles.
However, there is currently no general theory that could explain the empirical
success of BP (with the notable exception of the use of BP in LDPC decoding~\cite{LMSS98, LMSS01,RSU01}).

A striking recent application of BP is to instances of NP-hard constraint satisfaction problems such
as 3-SAT or 3-coloring; this is the type of problems that we are dealing with in the present work.
In this case the primary objective is \emph{not} to compute the marginals of some distribution,
but to construct a solution to the constraint satisfaction problem.
For example, BP can be used to (attempt to) compute a proper 3-coloring of a given graph.
Indeed, empirically BP (and its sibling Survey Propagation ``SP'')
seems to perform well on problem instances that are notoriously ``hard'' for other current algorithmic approaches,
including the case of sparse \emph{random graphs}.

For instance, let $G(n,p)$ be the random graph with vertex set $V=\{1,\ldots,n\}$ that is obtained
by including each possible edge with probability $0<p=p(n)<1$ independently.
Thus, the expected degree of any vertex in $G(n,p)$ is $(n-1)p\sim np$.
Then there exists a threshold $\tau=\tau(n)$ such that for any $\eps>0$ the random graph
$G(n,p)$ is $3$-colorable with probability $1-o(1)$
if $np<(1-\eps)\tau$, whereas $G(n,p)$ is not $3$-colorable if $np>(1+\eps)\tau$~\cite{AchFried}.
In fact, random graphs $G(n,p)$ with average degree $np$ just below $\tau$
were considered \emph{the} example of ``hard'' instances of the 3-coloring problem,
until statistical physicists discovered that BP/SP can solve these graph problems efficiently
in a regime considered ``hard'' for any previously known algorithms (possibly right up to the threshold density)
\cite{Weigt,BMPWZCol}.
While there are exciting and deep arguments from statistical physics that provide a plausible
explanation of why these message passing algorithms succeed, these arguments are non-rigorous, and indeed
no mathematically rigorous analysis is currently known.

The difficulty in understanding the performance of BP/SP on $G(n,p)$ actually lies in two aspects.
The first aspect is the combinatorial structure of the random graph $G(n,p)$ with respect to the 3-coloring problem,
which is not very well understood.
In fact, even the basic problem of obtaining the precise value of the threshold $\tau$ is one of the current challenges
in the theory of random graphs.
Furthermore, we lack a rigorous understanding of
the ``solution space geometry'', i.e., the structure of the set of all proper 3-colorings of a typical random graph $G(n,p)$
(e.g., how many proper 3-colorings are there typically, and what is the typical Hamming distance between any two).
But according to the statistical physics analysis,
the solution space geometry affects the behavior of BP significantly.

The second aspect, which we focus on in the present work, is the actual BP algorithm:
given a graph $G$, how/why does the BP algorithm ``construct'' a 3-coloring?
Thus far there has been no rigorous analysis of BP that applies to graph coloring
instances except for graphs that are \emph{globally} tree-like (such as trees or forests).
However, it seems empirically that BP performs well
on many graphs that are just \emph{locally} tree like (i.e., do not contain ``short'' cycles).
Therefore, in the present paper our goal is to analyze BP rigorously on a class
of graphs that may have a complex combinatorial structure globally, but that have a very simple solution space geometry.
More precisely, we shall relate the success of BP to spectral properties of the adjacency matrix of the input graph.
In addition, we point out that the analysis comprises a natural random graph model (namely, a ``planted solution'' model).

\subsection{Belief Propagation and Spectral Techniques}\label{Sec_BPSpec}

The main contribution of this paper is a rigorous analysis of BP for 3-coloring.
We basically show that if a certain (simple) spectral heuristic for 3-coloring succeeds, then so does BP.
Thus, the result does not refer to a specific random graph model,
but to a special class of graphs -- namely graphs that satisfy a certain spectral condition.
More precisely, we say that a graph $G=(V,E)$ on $n$ vertices is \emph{$(d,\eps)$-regular} if
there exists a 3-coloring of $G$ with color classes $V_1,V_2,V_3$ such that the following is true.
Let $\vecone_{V_i}\in\RR^V$ be the vector whose entries equal $1$ on coordinates $v\in V_i$,
and $0$ on all other coordinates; then
\begin{description}
\item[R1.] for all $1<i<j<3$ the vector $\vecone_{V_i}-\vecone_{V_j}$ is an eigenvector of the adjacency matrix $A(G)$
	with eigenvalue $-d$, and
\item[R2.] if $\xi\perp\vecone_{V_i}$ for all $i=1,2,3$, then $\|A(G)\xi\|\leq\eps d\|\xi\|$.
\end{description}
We shall state a few elementary properties of $(d,\eps)$-regular graphs in Proposition~\ref{Prop_dominant} below
(assuming that $\eps$ is sufficiently small -- $\eps<0.01$, say).
For instance, we shall see that $(d,\eps)$-regularity implies that each vertex $v\in V_i$ has precisely $d$
neighbors in each other color class $V_j$ ($i\not=j$).
Moreover, $(V_1,V_2,V_3)$ is the only 3-coloring of $G$ (up to permutations of the color classes, of course),
and for each pair $i\not=j$ the bipartite graph consisting of the $V_i$-$V_j$-edges is an expander.

Furthermore, if a graph $G$ is $(d,\eps)$-regular for any $\eps<0.01$, say, then the following spectral heuristic is easily seen to produce a 3-coloring.
\begin{enumerate}
\item Compute a pair of perpendicular eigenvectors $\chi^{1},\chi^{2}\in\RR^V$ of $A(G)$ with eigenvalue $-d$.
\item Define an equivalence relation $\approx$ on $V$ by letting
	$v\approx w$ iff $\chi_v^{i}=\chi_w^{i}$ for $i=1,2$.
	Output the equivalence classes of $\approx$ as a 3-coloring of $G$.
\end{enumerate}
The equivalence classes of $\approx$ are precisely the three color classes $V_1,V_2,V_3$.
For if $v,w$ belong to the same color class, then their entries in all three vectors $\vecone_{V_i}-\vecone_{V_j}$ ($i<j$) coincide;
hence, as the space spanned by these vectors contains $\chi^{1},\chi^{2}$, we have $v\approx w$.
Conversely, if $v\approx w$, then the entries of $v$ and $w$ in all the vectors $\vecone_{V_i}-\vecone_{V_j}$ coincide,
because these vectors lie in the space spanned by $\chi^{1},\chi^{2}$;
consequently, $v,w$ belong to the same color class $V_k$.

The main result of this paper is that BP can 3-color $(d,0.01)$-regular graphs in polynomial time,
provided that $d$ is not too small and the number of vertices is sufficiently large.
We defer the description of the actual (randomized, polynomial time) BP coloring algorithm \texttt{BPCol}, which the following theorem
refers to, to Section~\ref{Sec_BPCol}.

\begin{thm}\label{thm:BPforCol}
There exist constants $d_0,\kappa>0$ such that for each $d\geq d_0$ there is a number $n_0=n_0(d)$ so that the following holds.
If $G=(V,E)$ is a $(d,0.01)$-regular graph on $n=|V|\geq n_0$ vertices, then
with probability $\geq\kappa n^{-1}$ over the coin tosses of the algorithm,
\texttt{BPCol}$(G)$ outputs a proper 3-coloring of $G$.
\end{thm}
Observe that Theorem~\ref{thm:BPforCol} deals with ``sparse'' graphs, since the lower bound $n_0$ on the number
of vertices depends on $d$.
The proof yields an exponential dependence, i.e., $n_0=\exp(\Theta(d))$.
Conversely, this means that the average degree of $G$ is at most logarithmic in $n$,
which is arguably the most relevant regime to analyze BP (cf.~Section~\ref{Sec_BPCol}).
Moreover, by applying \texttt{BPCol} $O(n)$ times independently, the success probability can be boosted to $1-\alpha$
for any $\alpha>0$.
Besides, there is an easy way to modify the (initialization step of) \texttt{BPCol} so that
the success probability of one iteration is at least $\kappa$ (rather than $\kappa n^{-1}$), cf.\ Remark~\ref{Rem_Success} for details.

Let us emphasize that the contribution of Theorem~\ref{thm:BPforCol} is \emph{not} that we can now 3-color a class
of graphs for which no efficient algorithms were previously known, as
the aforementioned spectral heuristic 3-colors $(d,0.01)$-regular graphs in polynomial time.
Instead, the new aspect is that we can show that the \emph{Belief Propagation algorithm} 3-colors
$(d,0.01)$-regular instances, thus shedding new light on this 
heuristic.
Indeed, the proof of Theorem~\ref{thm:BPforCol}, which we present in Section~\ref{Sec_thm:BPforCol},
shows that in a sense \texttt{BPCol} ``emulates'' the spectral heuristic (although no spectral techniques
occur in the description of \texttt{BPCol}).
Thus, we establish a connection between spectral methods and BP.
Besides, we note that no ``purely combinatorial'' algorithm (that avoids the use of advanced techniques
such as Semidefinite Programming or spectral methods) is known to 3-color $(d,0.01)$-regular graphs.

To illustrate Theorem~\ref{thm:BPforCol}, and to provide an example of $(d,0.01)$-regular graphs,
we point out that the main result comprises a regular random graph model with a ``planted'' 3-coloring.
Let $G_{n,d,3}$ be the random graph with vertex set $V=\{1,\ldots,3n\}$ obtained as follows.
\begin{enumerate}
\item Let $V_1,V_2,V_3$ be a random partition of $V$ into three pairwise disjoint sets of equal size.
\item For any pair $1<i<j<3$ independently choose a $d$-regular bipartite graph
	with vertex set $V_i\du V_j$ uniformly at random.
\end{enumerate}
For a fixed $d$ we say that $G_{n,d,3}$ has a certain property $\mathcal{P}$
\emph{with high probability} (``w.h.p''),
if the probability that $G_{n,d,3}$ enjoys $\mathcal{P}$ tends to $1$ as $n\ra\infty$.
Concerning $G_{n,d,3}$, Theorem~\ref{thm:BPforCol} implies the following.

\begin{cor}\label{cor:BPforCol}
Suppose that $d\geq d_0$ is fixed.
With high probability a random graph $G=G_{n,d,3}$ has the following property:
with probability $\geq\kappa n^{-1}$ over the coin tosses of the algorithm,
$\texttt{BPCol}(G)$ outputs a proper 3-coloring of $G$.
\end{cor}
\noindent
To prove Corollary~\ref{cor:BPforCol}, we show that w.h.p.\ $G_{n,d,3}$ is $(d,0.01)$-regular,
cf.~Section~\ref{sec:cor_BPforCol}.



\subsection{Related Work} \label{subsec:related}

Alon and Kahale~\cite{AlonKahale97} were the first to employ spectral techniques for 3-coloring sparse random graphs.
They present a spectral heuristic and show that this heuristic finds a 3-coloring
in the so-called ``planted solution model''.
This model is somewhat more difficult to deal with algorithmically than the $G_{n,d,3}$ model that we study in the present work.
For while in the $G_{n,d,3}$-model each vertex $v\in V_i$ has \emph{exactly} $d$ neighbors in each of the
other color classes $V_j\not=V_i$, in the planted solution model of Alon and Kahale
the number of neighbors of $v\in V_i$ in $V_j$ has a Poisson distribution with mean $d$.
In effect, the spectral algorithm in~\cite{AlonKahale97} is more sophisticated than the spectral heuristic
from Section~\ref{Sec_BPSpec}.
In particular, the Alon-Kahale algorithm succeeds on $(d,0.01)$-regular graphs
(and hence on $G_{n,d,3}$ w.h.p.).

There are numerous papers on the performance of message passing algorithms
for constraint satisfaction problems
(e.g., Belief Propagation/Survey Propagation) by authors from the statistical physics community
(cf.~\cite{Weigt,BrMeZe:03,pnas} and the references therein).
While these papers provide rather plausible (and insightful) explanations for the success
of message passing algorithms on problem instances
such as  random graphs $G_{n,p}$ or random $k$-SAT formulae,
the arguments (e.g., the replica or the cavity method) are mathematically non-rigorous.
To the best of our knowledge, no connection between spectral methods and BP
has been established in the physics literature.

Feige, Mossel, and Vilenchik~\cite{WP} showed that the Warning Propagation (WP) algorithm for 3-SAT
converges in polynomial time to a satisfying assignment on a model of
random 3-SAT instances with a planted solution.
Since the messages in WP are additive in nature, and not multiplicative as in BP,
the WP algorithm is conceptually much simpler.
Moreover, on the model studied in~\cite{WP}  a fairly simple combinatorial algorithm
(based on the ``majority vote'' algorithm) is known to succeed.
By contrast, no purely combinatorial algorithm
(that does not rely on spectral methods or semi-definite programming)
is known to 3-color $G_{n,d,3}$ or even arbitrary $(d,0.01)$-regular instances.

A very recent  paper by Yamamoto and Watanabe~\cite{Osamu} deals with a
spectral approach to analyzing BP for the Minimum Bisection problem.
Their work is similar to ours in that they point out that a BP-related algorithm
\texttt{pseudo-bp} emulates spectral methods.
However, a significant difference is that  \texttt{pseudo-bp} is a 
simplified version of BP that is easier to analyze,
whereas in the present work we make a point of analyzing the BP algorithm
for coloring as it is stated in~\cite{Weigt}
(cf.\ Remark~\ref{Rem_Osamu} for more detailed comments).
Nonetheless, an interesting aspect of~\cite{Osamu} certainly is that this paper shows that BP
can be applied to an actual optimization problem,
rather than to the problem of just finding any feasible solution (e.g., a $k$-coloring).

The effectiveness of
message passing algorithms for amplifying local information in order
to decode codes close to channel capacity was recently established
in a number of papers, e.g.~\cite{LMSS98, LMSS01,RSU01}. Our results are
similar in flavor, however the analysis provided here
allows to recover a proper 3-coloring of the entire graph,
whereas in the random LDPC codes setting, message passing allows to
recover only a $1-o(1)$ fraction of the codeword correctly. In
\cite{LMSS01} it is shown that for the erasure channel, all bits may
be recovered correctly using a message passing algorithm, however in
this case the message passing algorithm is of combinatorial nature
(all messages are either $0$ or $1$) and the LDPC code is designed
so that message passing works for it.

\section{The Belief Propagation Algorithm for 3-Coloring}\label{Sec_BPCol}

\begin{figure}
\begin{center}
        \psfrag{v}{$w$}
        \psfrag{w}{$v$}
        \psfrag{Nwohnev}{$\hspace{-3mm}N(v)\setminus\{w\}$}
        \includegraphics[height=3cm]{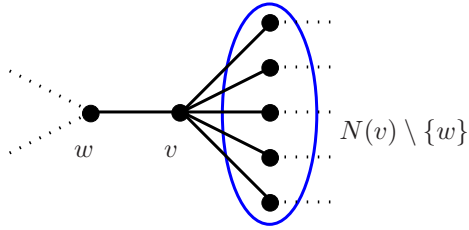}
\end{center}
\caption{the BP equation.}\label{Fig_BP}
\end{figure}

Following~\cite{Weigt}, in this section we will describe the basic ideas behind the BP algorithm.
Since BP is a heuristic based on non-rigorous ideas (mainly from artificial intelligence and/or statistical physics),
the discussion of its main ideas will lack mathematical rigor a bit; in fact, some of the assumptions
that BP is based on (e.g., ``asymptotic independence'') may seem ridiculous at first glance.
Nonetheless, as we pointed out in the introduction, BP makes up for this by being very successful empirically.
At the end of this section, we will state the version of BP that we are going to work with precisely.

The basic strategy behind the BP algorithm for 3-coloring is
to perform a fixed point iteration for certain ``messages'', starting from a suitable initial assignment.
In the case of 3-coloring the messages correspond to the edges of the graph and to the three available colors.
More precisely, to each (undirected) edge $\{v,w\}$ of the
graph $G=(V,E)$ and each color $a\in\{1,2,3\}$ we associate two messages
$\eta_{v\ra w}^a$ from $v$ to $w$ about $a$, and $\eta_{w\ra v}^a$ from $w$ to $v$ about $a$;
in general, we will have $\eta_{v\ra w}^a\not=\eta_{w\ra v}^a$.
Thus, the messages are \emph{directed} objects.
Each of these messages $\eta_{v\ra w}^a$ is a number between $0$ and $1$, which we interpret as the ``probability'' that
vertex $v$ takes the color $a$ in the graph obtained from $G$ by removing $w$.
Here ``probability'' refers to the choice of a random (proper) 3-coloring of $G-w$, while the graph $G$ is considered fixed.
(There is an obvious symmetry issue with this definition, which we will discuss shortly.)

Having introduced the variables $\eta_{v\ra w}^a$, we can set up the
\emph{Belief Propagation Equations} for coloring, which
are the basis of the BP algorithm.
The BP equations reflect a relationship that the probabilities $\eta_{v\ra w}^a$
should (approximately) satisfy under certain assumptions on the graph $G$, namely that
	\begin{equation}\label{eqBP}
	\eta_{v\ra w}^a=\frac{\prod_{u\in N(v)\setminus w}1-\eta_{u\ra v}^a}
		{\sum_{b=1}^3\prod_{u\in N(v)\setminus w}1-\eta_{u\ra v}^b}
	\end{equation}
for all edges $\{v,w\}$ of $G$ and all $a\in\{1,2,3\}$ (cf.~Figure~\ref{Fig_BP}).

The idea behind~(\ref{eqBP}) is that $v$ takes color $a$ in the graph $G-w$ iff none of
its neighbors $u\in N(v)\setminus w$ has color $a$ in $G-v$.
Furthermore, the probability of this event (``no $u$ has color $a$'') is assumed to be (asymptotically)
equal to the \emph{product} $\prod_{u\in N(v)\setminus w}1-\eta_{u\ra v}^a$ of the individual probabilities;
that is, the neighbors $u\not=w$ of $v$ are assumed to be \emph{asymptotically independent}.
Of course, this assumption 
does not hold for arbitrary graphs $G$. 
Finally, the numerator on the r.h.s.\ of~(\ref{eqBP}) is just a normalizing term, which ensures
that $\sum_{a=1}^3\eta_{v\ra w}^a=1$.

The reason why in the above discussion we refer to the probability that $v$ takes color $a$
in the graph $G-w$ \emph{obtained by removing} $w$ rather than just to the probability that $v$ takes color $a$ in $G$
is that in the latter case the neighbors $u\in N(v)$ would \emph{never} be (asymptotically) independent -- not even if $G$ is a tree.
For in this case the presence of $v$ -- more precisely, the existence of the short path $(u,v,u')$ for any two neighbors $u,u'\in N(v)$ of $v$ --
would render the colors within the neighborhood $N(v)$ heavily dependent.
Similarly, if $G$ contains triangles, so that for some vertices $v$ the neighborhood $N(v)$ is not an independent set,
then the independence assumption that is implicit in~(\ref{eqBP}) will be violated.
Nonetheless, if $G$ does not feature (many) short cycles -- say, all the cycles are of length
$\Omega(\log|V|)$ as $|V|\rightarrow\infty$ -- then the BP equations~(\ref{eqBP}) may at least be asymptotically valid.
The random graph model $G_{n,d,3}$ provides an example of graphs (essentially) without such short cycles.

Now, the basic idea behind the BP algorithm is the following.
We start with a ``reasonable'' initial assignment $\eta_{v\ra w}^a(0)$ and use~(\ref{eqBP}) to perform
a fixed point iteration by letting
	\begin{equation}\label{eqBPOp}
	\eta_{v\ra w}^a(l+1)=\frac{\prod_{u\in N(v)\setminus\{w\}}1-\eta^a_{u\ra v}(l)}
		{\sum_{b=1}^3\prod_{u\in N(v)\setminus\{w\}}1-\eta^b_{u\ra v}(l)}
	\end{equation}
for all $\{v,w\}\in E$ and $a\in\{1,2,3\}$.
As soon as some of the values $\eta_{v\ra w}^a(l+1)$ are strongly ``biased'' toward either $0$ or $1$,
we try to exploit this information to obtain a coloring.

Before we state the BP algorithm precisely, we need to discuss
an important issue with the BP equations~(\ref{eqBP}).
Namely, in the case of  3-coloring the set of all 3-colorings is symmetric under permuting the color classes.
Therefore, if we actually define $\eta_{v\ra w}^a$ to equal the probability w.r.t.\ a random 3-coloring of $G-w$,
then trivially $\eta_{v\ra w}^a=\frac13$ for all $a,v,w$.
In fact, this trivial solution is actually a fixed point of~(\ref{eqBPOp}).
Hence, we need to ``break symmetry''.
In particular, it is not a good idea to choose the initial assignment $\eta_{v\ra w}^a(0)=\frac13$ for all $a,v,w$.
Therefore, we do not start from $\eta_{v\ra w}^a(0)=\frac13$,
but we assign to each $\eta_{v\ra w}^a$ the value $\frac13$ plus a small random error $\delta$.
The hope is that this random error will cause the fixed point iterations~(\ref{eqBPOp}) to converge to a non trivial
fixed point (other than $\eta_{v\ra w}^a(0)=\frac13$ for all $a,v,w$), and that this fixed point
yields sufficient information to 3-color $G$.
For instance, if $\chi:V\ra\{1,2,3\}$ is a 3-coloring of $G$, then
	$$\eta_{v\ra w}^a=\left\{\begin{array}{cl}1&\mbox{ if }\chi(v)=a\\0&\mbox{ otherwise}\end{array}\right.
		\qquad(a=1,2,3;\,\{v,w\}\in E)$$
is a fixed point of~(\ref{eqBPOp}), and clearly the 3-coloring $\chi$ can be read out of the above messages easily.
The algorithm \texttt{BPCol} is shown in Fig.~\ref{Fig_BPCol}.
Observe that Step~1 ensures that
	\begin{eqnarray}\label{eqIniI}
	\sum_{a=1}^3\eta_{v\ra w}^a(0)&=&1\qquad\mbox{for all }\{v,w\}\in E.
	\end{eqnarray}


\begin{figure}
\begin{Algo}\label{Alg_Color}\upshape\texttt{BPCol$(G)$}\\\sloppy
\emph{Input:} A graph $G=(V,E)$.\
\emph{Output:} An assignment of colors to the vertices of $G$.
\begin{tabbing}
mmm\=mm\=mm\=mm\=mm\=mm\=mm\=mm\=mm\=\kill
{\algstyle1.}	\> \parbox[t]{38em}{\algstyle
			Let $\delta=\exp(-\log^{3}n)$.\\
			For each $v\in V$ perform the following independently:
			}\\
\> \> \parbox[t]{36em}{\algstyle
			choose $a\in\{1,2,3\}$ uniformly at random and assign $\eta_{v\rightarrow w}^a(0)=\frac13+\delta$
				and $\eta_{v\rightarrow w}^b(0)=\frac13-\frac{\delta}2$ for all $b\in\{1,2,3\}\setminus\{a\}$ and $w\in N(v)$.
			}\\
{\algstyle2.}	\> \parbox[t]{42em}{\algstyle For $l=1,\ldots,l^*=\lceil\log^4 n\rceil$}\\
\> \> \parbox[t]{40em}{\algstyle
		compute
			$\eta_{v\ra w}^a(l+1)$ using~(\ref{eqBPOp}) for all $a$, $v$, and $w$.}\\
{\algstyle3.}	\> \parbox[t]{42em}{\algstyle For each $v\in V$ and each $a\in\{1,2,3\}$
		compute $\beta^a_v=|N(v)|^{-1}\sum_{u\in N(v)}1-\eta^a_{u\ra v}(l^*)$.\\
		Assign to each $v\in V$ a color $a\in\{1,2,3\}$ such that $\beta^a_v=\max_{b\in\{1,2,3\}}\beta^b_v$.
		}
\end{tabbing}
\end{Algo}
\caption{the algorithm \texttt{BPCol}.}\label{Fig_BPCol}
\end{figure}

\begin{remark}\upshape\label{Rem_Success}
Theorem~\ref{thm:BPforCol} states that the probability (over the random decisions in Step~1)
that \texttt{BPCol} yields a proper 3-coloring of its $(d,0.01)$-regular input graph
is $\Omega(n^{-1})$.
This can be boosted to $\Omega(1)$ by means of the following slightly more careful initialization.
Instead of choosing a random $a$ for each $v\in V$ independently, we choose a random permutation
$\sigma$ of $V$ and let
	$W_a=\{\sigma((a-1)n/3+1),\ldots,\sigma(an/3)\}$ ($a=1,2,3$).
Then, for each $v\in W_a$ we set 
$\eta_{v\rightarrow w}^a(0)=\frac13+\delta$ and
$\eta_{v\rightarrow w}^b(0)=\frac13-\frac{\delta}2$ ($b\in\{1,2,3\}\setminus\{a\}$, $w\in N(v)$).
The proof of Proposition~\ref{Prop_initial} below shows that this leads to a success probability of $\Omega(1)$.
Nonetheless, we chose to state \texttt{BPCol} with independent decisions in its initalization,
because this appears more natural (and generic) to us.
\end{remark}

\begin{remark}\upshape
Although in the above discussion of the BP equation~(\ref{eqBPOp}) we referred to ``local'' properties (such as
the absence of short cycles), such local properties will not occur explicitly in our analysis of \texttt{BPCol}.
Indeed, relating \texttt{BPCol} to spectral graph properties, the analysis has a ``global'' character.
Nonetheless, various local conditions (e.g., a relatively small number of short cycles) are implicit in
the ``global'' assumption that the graph $G$ is $(d,0.01)$-regular (cf.~Theorem~\ref{thm:BPforCol}).
For more background on spectral vs.\ combinatorial graph properties cf.\ Chung and Graham~\cite{ChungGraham}. 
\end{remark}

\begin{remark}\upshape
\texttt{BPCol} updates the messages $\eta_{v\ra w}^a$ ``in parallel'', i.e, the messages carry ``time stamps'' (cf.~(\ref{eqBPOp})).
An alternative, equally common option would be ``serial'' updates, e.g., by choosing each time a random pair $v,w$ of adjacent vertices
along with a color $a\in\{1,2,3\}$ and updating $\eta_{v\ra w}^a$ via~(\ref{eqBP}).
\end{remark}

\begin{remark}\upshape
\texttt{BPCol} exploits the result of the fixed point iteration~(\ref{eqBPOp}) in a more straightforward fashion than
the version of BP described in~\cite{Weigt}.
Namely, after performing a fixed point iteration of~(\ref{eqBPOp}),
the algorithm in~\cite{Weigt} does not assign colors to \emph{all} vertices (as Step~3 of \texttt{BPCol} does),
but only to a small fraction (the most decisive ones with respect to the calculated values).
Then, the algorithm performs another fixed point iteration, etc.
The reason is that in the
random graph model considered in~Ê\cite{Weigt} typically the number of proper 3-colorings is exponential in the number of vertices,
whereas $(d,0.01)$-regular graphs have only one 3-coloring (up to permutations of the colors).
\end{remark}

\begin{remark}\upshape\label{Rem_Osamu}
Let us discuss the essential differences between \texttt{BPCol} for $k=2$ and the algorithm
\texttt{pseudo-bp} analyzed in~\cite{Osamu}.
\begin{enumerate}
\item In \texttt{pseudo-bp} the products in~(\ref{eqBP}) are taken over \emph{all} neighbors of $v$, including $w$.
	This apparently minor modification has a major impact on the analysis.
	For including $w$ causes the messages $\eta_{v\ra w}^a$ to be independent of $w$.
	Consequently, in \texttt{pseudo-bp} the messages at time $l$ are $2|V|$-dimensional objects, whereas in the present
	work the dimension is $2k|E|$.
\item \texttt{pseudo-bp} actually works with the logarithms $\ln(\eta_{v\ra w}^a)$ of the messages instead of the original $\eta_{v\ra w}^a$.
	Of course, the equation~(\ref{eqBP}) can be phrased in terms of $\ln(\eta_{v\ra w}^a)$
	as $\ln(\eta_{v\ra w}^a)=F(\ln(\eta_{u\ra v}^a))_{u\in N(v)}$ for some function $F$.
	Now, in \texttt{pseudo-bp} this non-linear function $F$ is replaced by a truncated linear function $\hat F$.
\end{enumerate}
\end{remark}

\section{Proof of Theorem \ref{thm:BPforCol}}\label{Sec_thm:BPforCol}

\subsection{Preliminaries and Notation}

Throughout this section, we let $\eps>0$ be a sufficiently small constant (whose value will be determined implicitly in the course
of the proof).
Moreover, we keep the assumptions from Theorem~\ref{thm:BPforCol}.
Thus, we let $d>d_0$ for a sufficiently large constant $d_0$; in particular, we assume that $d_0>\exp(\eps^{-2})$.
In addition, we assume that $n>n_0$ for some sufficiently large number $n_0=n_0(d)$,
and that $G=(V,E)$ is a $(d,0.01)$-regular graph on $n=|V|$ vertices.
This is reflected by the use of asymptotic notation in the analysis, which always refers to $n$ being sufficiently large.

Furthermore, we let $(V_1,V_2,V_3)$ be a 3-coloring of $G$ with respect to which the conditions {\bf R1} and {\bf R2}
from the definition of $(d,0.01)$-regularity hold.
(Actually a $(d,0.01)$-regular graph has a unique 3-coloring up to permutations of the color classes, but
we will not use this fact.)
The following easy observation will be used frequently.
\begin{lem}\label{Lemma_reg}
Let $i,j\in\{1,2,3\}$, $i\not=j$.
Then in $G$ each vertex $v\in V_i$ has precisely $d$ neighbors in $V_j$.
Consequently, $|N(v)|=2d$.
\end{lem}
\begin{proof}
Assume w.l.o.g.\ that $i=1$ and $j=2$.
By condition {\bf R1} $\xi=\vecone_{V_i}-\vecone_{V_j}$ is an eigenvector of the adjacency matrix
$A(G)=(a_{vw})_{v,w\in V}$ with eigenvalue $-d$.
Hence, letting $\eta=-d\xi=A(G)\xi$, we have
	$-d=\eta_v=-\sum_{w\in N(v)\cap V_j}a_{vw}=-|N(v)\cap V_j|.$
\end{proof}

Following~\cite{Weigt}, we will denote the elements $(v,w)\in\AAA$ as $v\ra w$.
Furthermore, we shall frequently work with the vector space $\RRR=\RR^3\tensor\RR^{\AAA}$.
Each element $\Gamma\in\RRR$ has a unique representation
	$$\Gamma=\bc{\begin{array}{c}1\\0\\0\end{array}}\tensor\Gamma^1
		+\bc{\begin{array}{c}0\\1\\0\end{array}}\tensor\Gamma^2
		+\bc{\begin{array}{c}0\\0\\1\end{array}}\tensor\Gamma^3$$
with $\Gamma^i=(\Gamma_{v\ra w}^i)_{v\ra w\in\AAA}\in\RR^{\AAA}$ ($i=1,2,3$).
Hence, we shall denote such a vector as $\Gamma=(\Gamma_{v\ra w}^i)_{v\ra w\in\AAA,i\in\{1,2,3\}}$.
Semantically, one can think of $\Gamma_{v\ra w}^i$ as the ``message'' that $v$ sends to $w$ about color $i$.
Note that the messages $\eta_{v\ra w}^a(l)$ defined from Section~\ref{Sec_BPCol} constitute
vectors $\eta(l)=(\eta_{v\ra w}^a(l))_{v\ra w\in\AAA,a\in\{1,2,3\}}\in\RRR$.

We will denote the scalar product of vectors $\xi,\eta$ as $\scal{\xi}{\eta}$.
Moreover, $\|\xi\|=\sqrt{\scal{\xi}{\xi}}$ denotes the $\ell_2$-norm.
In addition, if $M:\RR^{n_1}\ra\RR^{n_2}$ is linear, then we let
$\|M\|=\max_{\xi\in\RR^{n_1},\,\|\xi\|=1}\|M\xi\|$ signify the operator norm of $M$.
Further, $M^T$ denotes the transpose of $M$, i.e., the unique linear operator
$\RR^{n_2}\ra\RR^{n_1}$ such that
$\scal{M\xi}{\eta}=\scal{\xi}{M^T\eta}$ for all $\xi\in\RR^{n_1}$, $\eta\in\RR^{n_2}$.

\subsection{Outline of the Analysis}

In order to analyze \texttt{BPCol}, we shall relate the fixed point iteration of (\ref{eqBPOp}) to the spectral coloring algorithm
from Section~\ref{Sec_BPSpec}.
More precisely, we will approximate the fixed point iteration of the non-linear operation~(\ref{eqBPOp})
by a fixed point iteration for a linear operator. 
One of the key ingredients in the analysis is to show how symmetry is broken (i.e., convergence to the all-$\frac13$ fixed point is avoided).
Indeed, it may not be clear \emph{a priori} that this will happen at all,
because the random bias generated in Step~1 of \texttt{BPCol} is uncorrelated to the planted coloring.
The analysis is based on the following crucial observation (cf.\ Corollary \ref{Cor_Startgame} below):
after a logarithmic number of iterations,
for all $v\in V_i,w\in V_j,i\neq j$ the messages $\eta_{v\ra w}^a$ are dominated by eigenvectors of the linear operator
which we use to approximate~(\ref{eqBPOp}).
Furthermore, these eigenvectors
mirror the coloring $(V_1,V_2,V_3)$ and
are (almost) constant on every color class $V_i$ (with basically $0,1,-1$ values on the different color classes).
Hence, the (random) initial bias gets amplified so that 
the planted 3-coloring can eventually be read out of the messages.

To carry out this analysis precisely, we set
	$$\Delta_{v\ra w}^a(l)=\eta_{v\ra w}^a(l)-\frac13.$$
Moreover, we let $\BB:\RRR\ra\RRR$ denote the (non-linear) operator defined by
	$$(\BB\Gamma)_{v\ra w}^a=-\frac13+\frac{\prod_{u\in N(v)\setminus w}1-\frac32\Gamma^a_{u\ra v}}
				{\sum_{b=1}^3\prod_{u\in N(v)\setminus w}1-\frac32\Gamma^b_{u\ra v}}\qquad(\Gamma\in\RRR).$$
Then~(\ref{eqBPOp}) can be rephrased in terms of the vectors
$\Delta(l)=(\Delta_{v\ra w}^a(l))_{v\ra w\in\AAA,\,a\in\{1,2,3\}}\in\RRR$ as
	\begin{equation}\label{eqBPOpDelta}
	\Delta(l+1)=\BB \Delta(l).
	\end{equation}
We shall see that we can approximate the non-linear operator $\BB$ in~(\ref{eqBPOpDelta}) by the following \emph{linear} operator
$\BB'$ if $\|\Delta(l)\|_{\infty}$ is small;
the operator $\BB'$ maps a vector $\Gamma=(\Gamma_{v\ra w}^a)_{a\in\{1,2,3\},v\to w\in \AAA}\in\RRR$ to the vector
$\BB'(\Gamma)=(\BB'(\Gamma)_{v\ra w}^a)_{a,v\to w}\in\RRR$ with entries
	\begin{equation}\label{eqBprime}
	\BB'(\Gamma)_{v\ra w}^a=-\frac12\sum_{u\in N(v)\setminus w}
		\Gamma_{u\ra v}^a+\frac16\sum_{b=1}^3\sum_{u\in N(v)\setminus w}\Gamma_{u \ra v}^b.
	\end{equation}
Indeed, $\BB':\RRR\ra\RRR$ is just the total derivative of $\BB$ at $0$.

We define a sequence $\Xi(l)$ by letting $\Xi(0)=\Delta(0)$ and
$\Xi(l)={\BB'}^l\Xi(0)$ for $l\geq1$, thinking of
$\Xi(l)$ as a ``linear approximation'' to $\Delta(l)$.
As a first step, we shall simplify the operator $\BB'$ a little.

\begin{lem}\label{Lemma_Bprime}
We have $(\BB'(\Xi(l)))_{v\ra w}^a=-\frac12\sum_{u\in N(v)\setminus w} \Xi_{u\ra v}^a(l)$ for all $l\geq0$,
	$v\ra w\in\AAA$, $a\in\{1,2,3\}$.
\end{lem}
\begin{proof}
Step~1 of \texttt{BPCol} ensures that the initial vector satisfies
	$$\sum_{b=1}^3 \Xi_{u \ra v}^b(0) =\sum_{b=1}^3 \Delta_{u \ra v}^b(0) = 0\qquad\mbox{for all }\{u,v\}\in E\quad(\mbox{cf.~(\ref{eqIniI})}).$$
Therefore, by induction and by the definition~(\ref{eqBprime}) of $\BB'$ we see that
	$\sum_{b=1}^3 \Xi_{u \ra v}^b(l)=0$ for all $l\geq0$.
Consequently, $\sum_{b=1}^3\sum_{u\in N(v)\setminus w}\Xi_{u \ra v}^b(l)=0$ for all $l\geq0$,
i.e., the second summand on the r.h.s.\ of~(\ref{eqBprime}) vanishes.
\end{proof}
\noindent
Due to Lemma~\ref{Lemma_Bprime}, we may just replace $\BB'$ by the simpler linear
operator $\LL:\RRR\ra\RRR$ defined by
	\begin{equation}\label{eqDefL}
	(\LL\Gamma)_{v\ra w}^a=-\frac12\sum_{u\in N(v)\setminus w}\Gamma_{u\ra v}^a\qquad(v\ra w\in\AAA,a\in\{1,2,3\}),
	\end{equation}
which satisfies
	\begin{equation}\label{eqXiLL}
	\Xi(l)=\LL^l\Xi(0)=\LL^l\Delta(0).
	\end{equation}
We also note for future reference that
	\begin{equation}\label{eqXiLLsum}
	\sum_{a=1}^3\Xi_{v\ra w}^a(l)=0\qquad\mbox{for all }v\ra w\in\AAA,\,l\geq0,
	\end{equation}
because~(\ref{eqIniI}) entails that (\ref{eqXiLLsum}) is true for $l=0$, whence the definition~(\ref{eqDefL}) of $\LL$ shows
that~(\ref{eqXiLLsum}) holds for all $l>0$.

In order to prove Theorem~\ref{thm:BPforCol}, we shall first analyze the sequence $\Xi(l)$ and then bound
the error $\|\Xi(l)-\Delta(l)\|_\infty$ resulting from linearization.
To study the sequence $\Xi(l)$, we investigate the dominant eigenvalues of $\LL$ and their corresponding eigenvectors.
More precisely, we shall see that our assumption on the spectrum of the adjacency matrix $A(G)$
implies that the dominant eigenvectors of $\LL$ mirror a 3-coloring of $G$.
We defer the proof of the following proposition to Section~\ref{Sec_dominant}.

\begin{proposition}\label{Prop_dominant}
Let $e_{ij}^a\in\RRR$ be the vector with entries
	\begin{eqnarray*}
	(e_{ij}^a)_{v\ra w}^b&=&\left\{
			\begin{array}{cl}
			1&\mbox{ if $b=a$, $v\in V_i$, and $w\in N(v)\cap V_j$},\\
			0&\mbox{ otherwise}
			\end{array}\ (v\ra w\in\AAA,a,b,i,j\in\{1,2,3\},i\not=j).
			\right.
	\end{eqnarray*}
Moreover, let $\EE$ be the space spanned by the 18 vectors $e_{ij}^a$ ($a,i,j\in\{1,2,3\},i\not=j$).
Then $\LL$ operates on $\EE$ as follows.
\begin{description}
\item[S1.] There are precisely six linearly independent eigenvectors $\{\zeta_{2}^a ,\zeta_{3}^a:a=1,2,3\}$ with eigenvalue
	$\lambda=\frac{d}4+\frac14\sqrt{d^2-8d+4}$, which satisfy
	\begin{equation}\label{eqzeta}
	\|\zeta_{2}^a-e_{12}^a-e_{13}^a+e_{21}^a+e_{23}^a\|_\infty\leq100d^{-1},\quad
		\|\zeta_{3}^a-e_{12}^a-e_{13}^a+e_{31}^a+e_{32}^a\|_\infty\leq100d^{-1}.
	\end{equation}
	These eigenvectors are symmetric with respect to the colors $a=1,2,3$, i.e., for any two distinct $a,b\in\{1,2,3\}$ and all
		$v\ra w\in\AAA$ we have
	\begin{equation}\label{eqzetasymm}
	(\zeta_j^a)_{v\ra w}^a=(\zeta_j^b)_{v\ra w}^b\mbox{, and }(\zeta_j^a)_{v\ra w}^b=0.
	\end{equation}
	In addition,
		\begin{equation}\label{eqzetanorm}
		\|\zeta_2^1\|=\|\zeta_j^a\|\mbox{ for all }j\in\{2,3\},a\in\{1,2,3\}.
		\end{equation}
\item[S2.] The three vectors $e^a=\sum_{i\not=j}e_{ij}^a$ with $a=1,2,3$ are eigenvectors
	with eigenvalue $\frac12-d$.
\item[S3.] For all $\xi\in\EE$ such that
	$\xi\perp\{e^a,\zeta_j^a:a=1,2,3,j=2,3\}$ we have
		$\|\LL\xi\|\leq\frac12\|\xi\|$.
\item[S4.] Furthermore, $\LL\EE\subset\EE$ and $\LL^T\EE\subset\EE$.
\end{description}
Finally, we have
\begin{description}
\item[S5.] $\|\LL^2\xi\|\leq0.01d^2\|\xi\|\quad\mbox{for all }\xi\perp\EE$.
\end{description}
\end{proposition}

The eigenvectors that we are mostly interested in are $\zeta_2^a,\zeta_3^a$ ($a=1,2,3$)
as~(\ref{eqzeta}) shows that these vectors represent the coloring $(V_1,V_2,V_3)$ completely.
As a next step, we shall show that $\Xi(l)$ can  be approximated well by a linear combination
of the vectors $\zeta_2^a,\zeta_3^a$, provided that $l$ is sufficiently large.
To this end, let
	\begin{equation}\label{eqxia}
	x_i^a=\sqrt{n}\cdot\frac{\scal{\Delta(0)}{\zeta_i^a}}{\|\Delta(0)\|\cdot\|\zeta_i^a\|}\qquad(i=2,3,\ a=1,2,3)
	\end{equation}
be the projection of the initial vector $\Delta(0)=\Xi(0)$ onto the eigenvector $\zeta_i^a$;
we shall see below that
the normalization in~(\ref{eqxia}) ensures that $x_i^a$ is bounded away from $0$.
Furthermore, recalling from~(\ref{eqzetanorm}) that $\|\zeta_i^a\|=\|\zeta_2^1\|$ for all $i,a$, we set
	\begin{equation}\label{eqnu}
	\nu=\frac{\|\Delta(0)\|}{\sqrt{n}\|\zeta_2^1\|}.
	\end{equation}

\begin{cor}\label{Cor_Startgame}
Suppose that $l\geq L_1= 2\lceil\log n\rceil$, and that $\Xi(0)\perp e^a$ for $a=1,2,3$.
Then
	$$\Xi^a_{v\ra w}(l)=\nu\lambda^{l}\sum_{a=1}^3\sum_{i=2}^3(x_i^a+o(1)){\zeta_i^a}_{v\ra w}
		\quad\mbox{for all $a\in\{1,2,3\}$ and $\{v,w\}\in E$.}$$
\end{cor}
\begin{proof}
Since by assumption the initial vector $\Xi(0)$ is perpendicular to $e^a$ for $a=1,2,3$ and
because $e^1,e^2,e^3$ are eigenvectors of $\LL$ by {\bf S2}, we have $\Xi(l)\perp e^a$.
Therefore, we can decompose $\Xi(l)$ as
	\begin{equation}\label{eqStartDecomp}
	\Xi(l)=\xi(l)+\sum_{a=1}^3\sum_{i=2}^3z_i^a(l)\zeta_i^a,\mbox{ where $\xi(l)\perp\{e^a,\zeta_i^a:i\in\{2,3\},a\in\{1,2,3\}\}$.}
	\end{equation}
Thus, to prove the corollary we need to compute the numbers $z_i^a(l)$ and bound $\|\xi(l)\|_\infty$.

With respect to the coefficients $z_i^a(l)$, note that $z_i^a(l)=\lambda^lz_i^a(0)$, because by {\bf S1} $\zeta_i^a$ is an eigenvector with
eigenvalue $\lambda$.
Moreover, $z_i^a(0)=\|\zeta_i^a\|^{-2}\scal{\Xi(0)}{\zeta_i^a}$.
Hence, (\ref{eqxia}) and~(\ref{eqnu}) yield $z_i^a(0)=x_i^a\cdot\nu$.
Thus,
	\begin{equation}\label{eqzia}
	z_i^a(l)=\lambda^l\nu\cdot x_i^a.
	\end{equation}

To bound the ``error term'' $\|\xi(l)\|_\infty$, we note that {\bf S3}--{\bf S5} entail
	\begin{equation}\label{eqLLsquared}
	\|\LL^2\gamma\|\leq0.01d^2\|\gamma\|\leq (0.3\lambda)^2\|\gamma\|
		\mbox{ for all }\gamma\perp\{e^a,\zeta_i^a:i\in\{2,3\},a\in\{1,2,3\}\},
	\end{equation}
provided that $d\geq d_0$ for a large enough constant $d_0>0$.
Let $k=\lfloor l/2\rfloor$.
Since $\xi(2k)=\LL^{2k}\xi(0)$, (\ref{eqLLsquared}) implies that
	\begin{eqnarray}\label{eqLLsquaredA}
	\|\xi(2k)\|&=&\|\LL^{2k}\xi(0)\|\leq(0.3\lambda)^{2k}\|\xi(0)\|\leq(0.3\lambda)^{2k}\|\Xi(0)\|.
	\end{eqnarray}
Moreover, as $l\leq2k+1$ and $\|\LL\|\leq d-\frac12$ by Proposition~\ref{Prop_dominant}, (\ref{eqLLsquaredA}) yields
	\begin{equation}\label{eqsmallxi}
	\|\xi(l)\|_{\infty}\leq\|\xi(l)\|\leq d\|\xi(2k)\|\leq d(0.3\lambda)^{l}\|\Xi(0)\|.
	\end{equation}
Finally, if $l\geq L_1$, then $d(0.3\lambda)^{l}\|\Xi(0)\|=o(\lambda^l\nu)$.
Thus, the assertion follows from~(\ref{eqStartDecomp}), (\ref{eqzia}), and~(\ref{eqsmallxi}).
\end{proof}

While in the initial vector $\Delta(0)=\Xi(0)$ the messages are completely uncorrelated with the coloring $(V_1,V_2,V_3)$,
Corollary~\ref{Cor_Startgame} entails that the dominant contribution to $\Xi(L_1)$ comes from the eigenvectors $\zeta_i^a$,
which represent that coloring.
This implies that all vertices $v$ in each class $V_a$ send essentially the same messages to all other
vertices $w\in V_b$ about each of the colors $1,2,3$, and these messages are solely determined by the initial projections $x_i^a$
of $\Delta(0)$ onto $\zeta_i^a$.
Hence, after $L_1$ iterations the messages are essentially coherent and strongly correlated to the planted coloring.
Thus, as a next step we analyze the distribution of the projections $x_i^a$.
To simplify the expression resulting from Corollary~\ref{Cor_Startgame}, let
	\begin{equation}\label{eqyia}
	y_1^a=x_2^a+x_3^a,\ y_2^a=-x_2^a,\ \mbox{and }y_3^a=-x_3^a.
	\end{equation}
Then (\ref{eqzeta}) and Corollary~\ref{Cor_Startgame} entail that
for all $v\in V_i$, all $w\in N(v)$, and $l\geq L_1$ we have 
	$$\Xi_{v\ra w}^a(l)=(y_i^a+o(1))\cdot\nu\lambda^l.$$
Of course, the numbers $y_i^a$ only depend on the initial vector $\Delta(0)$.
Therefore, we say that \emph{$\Delta(0)$ is feasible} if
\begin{description}
\item[F1.] $\Delta(0)\perp e^a$ for $a=1,2,3$, and
\item[F2.]	for any pair $a,b\in\{1,2,3\}$, $a\not=b$ we have
	\begin{equation}\label{eq_initial}
	|y_a^a-1|<\exp(-1/\epsilon)\mbox{ and }|y_a^b+0.5|<\exp(-1/\epsilon).
	\end{equation}
\end{description}

\begin{proposition}\label{Prop_initial}
With probability $\Omega(n^{-1})$ over the random bits used in Step~1 of \texttt{BPCol}
$\Delta(0)$ is feasible.
\end{proposition}
The elementary (though tedious) proof of Proposition~\ref{Prop_initial} can be found in Section~\ref{Sec_initial}.
Combining Corollary~\ref{Cor_Startgame} and~Proposition~\ref{Prop_initial}, we conclude that
with probability $\Omega(n^{-1})$ (namely, if $\Delta(0)$ is feasible) we have
	\begin{equation}\label{eqXiInfty}
	0.49\nu\lambda^l\leq\|\Xi(l)\|_\infty\leq1.1\nu\lambda^l\qquad(l\geq L_1).
	\end{equation}

Having obtained a sufficient understanding of the sequence $\Xi(l)$, we will now show that these vectors
provide a good approximation to the vectors $\Delta(l)$, which we are actually interested in.
The proof of the following proposition can be found in Section~\ref{Sec_err}.

\begin{proposition}\label{Prop_err}
Suppose that $\Delta(0)$ is feasible.
Let $L_2>0$ be the maximum integer such that $\|\Xi(L_2)\|_\infty\leq\eps$.
Then
	$\|\Xi(L_2)-\Delta(L_2)\|_\infty\leq-\log(\eps)\cdot\|\Xi(L_2)\|_\infty^2$.
\end{proposition}

Combining the information on the sequence $\Xi(l)$ provided by Corollary~\ref{Cor_Startgame} and Proposition~\ref{Prop_initial}
with the bound on $\|\Xi(L_2)-\Delta(L_2)\|_\infty$ from Proposition~\ref{Prop_err}, we can show that
the messages obtained in the next one or two steps of the algorithm already represent the coloring rather well.
To be precise, let us call the vector $\eta(l)$ \emph{proper} if
	$$\forall a\in\{1,2,3\},\,b\in\{1,2,3\}\setminus\{a\},\,v\in V_a,\,w\in N(v):\eta_{v\ra w}^a(l)\geq0.99\wedge\eta_{v\ra w}^b(l)\leq0.01.$$

\begin{proposition}\label{Prop_Endgame}
If $\Delta(0)$ is feasible, then for either $L_3=L_2+1$ or $L_3=L_2+2$ the vector $\eta(L_3)$ is proper.
\end{proposition}
\noindent
The proof of Proposition~\ref{Prop_Endgame} is the content of Section~\ref{Sec_Endgame}.

Proposition~\ref{Prop_Endgame} shows that the ``rounding procedure'' in
Step~3 of \texttt{BPCol} applied to the messages $\eta(L_3)$ would yield the coloring $(V_1,V_2,V_3)$.
However, \texttt{BPCol} actually applies that rounding procedure to $\eta(l^*)$, where $l^*>L_3$.
Therefore, in order to show that \texttt{BPCol} outputs a proper 3-coloring, we need to show that
these messages $\eta(l^*)$ are proper, too.

\begin{lem}\label{Lemma_runaway}
If $\eta(l)$ is proper, then so is $\eta(l+1)$.
\end{lem}
\begin{proof}
Let $v\in V_a$ for some $1\leq a\leq 3$, $w\in N(v)$, and $\{b,c\}=\{1,2,3\}\setminus\{a\}$.
Since $\eta(l)$ is proper, we have
	\begin{eqnarray}\label{eqEndgame4}
	\prod_{u\in V_c\cap N(v)\setminus w}\frac{1-\eta_{u\ra v}^a(l)}{1-\eta_{u\ra v}^b(l)}
		&\geq&\prod_{u\in V_c\cap N(v)\setminus w}1-\eta_{u\ra v}^a(l)\geq0.99^{2d},\\
	\prod_{u\in V_b\cap N(v)\setminus w}\frac{1-\eta_{u\ra v}^a(l)}{1-\eta_{u\ra v}^b(l)}
		&\geq&\left( \frac{0.99}{0.01} \right)^{2d-1}=99^{2d-1}.
	\end{eqnarray}
Consequently, the definition~(\ref{eqBPOp}) of the sequence $\eta(l)$ shows that
\begin{eqnarray}\nonumber
\frac{\eta_{v\ra w}^a(l+1)}{\eta_{v\ra w}^b(l+1)}&=&\prod_{u\in V_c\cap N(v)\setminus w}\frac{1-\eta_{u\ra v}^a(l)}{1-\eta_{u\ra v}^b(l)}\cdot
		\prod_{u\in V_b\cap N(v)\setminus w}\frac{1-\eta_{u\ra v}^a(l)}{1-\eta_{u\ra v}^b(l)}\\
	&\geq&
	0.01\cdot\left(\frac{(0.99)^2}{0.01}\right)^{2d} \geq 0.01\cdot90^{2d}\geq1000.
	\label{eqRunaway}
\end{eqnarray}
As the construction~(\ref{eqBPOp}) of $\eta(l+1)$ ensures that
$\eta_{v\ra w}^1(l+1)+\eta_{v\ra w}^2(l+1)+\eta_{v\ra w}^3(l+1)=1$,
(\ref{eqRunaway}) entails that $\eta_{v\ra w}^a(l+1)\geq0.99$ and $\eta_{v\ra w}^b(l+1)\leq0.01$,
whence $\eta(l+1)$ is proper.
\end{proof}

\noindent\emph{Proof of Theorem~\ref{thm:BPforCol}.}
Proposition~\ref{Prop_initial} states that $\Delta(0)$ is feasible with probability $\Omega(n^{-1})$.
Therefore, to establish Theorem~\ref{thm:BPforCol}, we show that \texttt{BPCol} outputs the coloring $(V_1,V_2,V_3)$
if $\Delta(0)$ is feasible.

Thus, assume that $\Delta(0)$ is feasible and let $L_2$ be the maximum integer such that $\|\Xi(L_2)\|_\infty\leq\eps$.
Then Corollary~\ref{Cor_Startgame} implies that $L_2=\Theta(\log^3n)$, because $\|\Xi(0)\|_\infty=\delta=\exp(-\log^3n)$,
and the $\ell_\infty$-norm of $\Xi(l)$ grows by a factor of $\lambda$ in each iteration.
Therefore, Proposition~\ref{Prop_Endgame} entails that $\eta(L_3)$ is proper for some $L_3=\Theta(\log^3n)$.
Thus, by Lemma~\ref{Lemma_runaway} the final $\eta(\ell^*)$ generated in Step~2 is proper,
whence Step~3 of \texttt{BPCol} outputs the coloring $V_1,V_2,V_3$.
\qed

\subsection{Proof of Proposition~\ref{Prop_dominant}}\label{Sec_dominant}

The operation~(\ref{eqDefL}) of $\LL$ is symmetric with respect to the three colors $a=1,2,3$.
Therefore, we shall represent $\LL$ as a tensor product of a $3\times 3$ matrix and an operator that represents the graph $G$.
To this end, we define operators $\MM:\RR^\AAA\ra\RR^\AAA$ and $\KK:\RR^\AAA\ra\RR^\AAA$ by
	\begin{equation}\label{eqDefMMKK}
	(\MM\Xi)_{v\ra w}=\sum_{u\in N(v)}\Xi_{u\ra v},\qquad
			(\KK\Xi)_{v\ra w}=\Xi_{w\ra v}\qquad(\Xi\in\RR^\AAA).
	\end{equation}
Thus,
	$$-\frac12((\MM-\KK)\Xi)_{v\ra w}=-\frac12\sum_{u\in N(v)\setminus w}\Xi_{u\ra v},$$
i.e., $-\frac12(\MM-\KK)$ represents the operation of $\LL$ with respect to a single color $a\in\{1,2,3\}$.
Therefore, we can rephrase the definition~(\ref{eqDefL}) of $\LL$ on the space $\RRR=\RR^3\tensor\RR^\AAA$ as
	\begin{equation}\label{eqStartgame2}
	\LL=-\frac12\id\tensor(\MM-\KK).
	\end{equation}
Hence, in order to understand $\LL$, we basically need to analyze $\MM-\KK$.

For $i,j\in\{1,2,3\}$ we define vectors $e_{ij}\in\RR^\AAA$ by letting
	$$(e_{ij})_{v\ra w}=\left\{\begin{array}{cl}
		1&\mbox{ if $v\in V_i$, $w\in V_j$, and $w\in N(v)$,}\\
		0&\mbox{ otherwise.}
		\end{array}\right.$$
The following lemma shows that it makes sense to split the analysis of $\MM-\KK$ into two parts:
first we shall analyze how $\MM-\KK$ operates on the space $\EE_0$ spanned by the vectors $e_{ij}$ ($1\leq i,j\leq 3$, $i\not=j$);
then, we will study the operation of $\MM-\KK$ on $\EE_0^{\perp}$.
\begin{lem}\label{Lemma_MMT}
If $\xi\in \EE_0$, then $\MM \xi,\MM^T\xi,\KK\xi,\KK^T\xi\in \EE_0$.
\end{lem}
\begin{proof}
Let $i,j,k\in\{1,2,3\}$ be pairwise distinct.
Since $\KK e_{ij}=e_{ji}$, we have $\KK \EE_0\subset \EE_0$.
Moreover, $\KK^T=\KK$.
Furthermore, by Lemma~\ref{Lemma_reg}
	\begin{equation}\label{eqMMKK1}
	(\MM e_{ij})_{v\ra w}=\sum_{u\in N(v)}(e_{ij})_{u\ra v}
		=\left\{\begin{array}{cl}d&\mbox{ if $v\in V_j$,}\\0&\mbox{ otherwise.}\end{array}\right.
	\end{equation}
Hence, $\MM e_{ij}=d(e_{jk}+e_{ji})$, and thus $\MM \EE_0\subset \EE_0$.
In addition, the transpose of $\MM$ is given by
	$$(\MM^T\Xi)_{v\ra w}=\sum_{u\in N(w)}\Xi_{w\ra u}.$$
Therefore,
	$$(\MM^T e_{ij})_{v\ra w}=\sum_{u\in N(w)}(e_{ij})_{w\ra u}
		=\left\{\begin{array}{cl}d&\mbox{ if $v\in V_i$,}\\0&\mbox{ otherwise.}\end{array}\right.$$
Consequently, $\MM^Te_{ij}=d(e_{ij}+e_{ik})$, whence $\MM^T\EE_0\subset \EE_0$.
\end{proof}

To study the operation of $\MM-\KK$ on $\EE_0$, note that~(\ref{eqMMKK1}) implies that
$(\MM-\KK)e_{ij}=de_{jk}+(d-1)e_{ji}$, if $i,j,k\in\{1,2,3\}$ are pairwise distinct.
Therefore, with respect to the basis $e_{12},e_{23},e_{21},e_{23},e_{31},e_{32}$ of $\EE_0$,
we can represent the operation of $\MM-\KK$ on $\EE_0$ by the $6\times 6$ matrix
	$$M=\bc{\begin{array}{cccccc}0&0&d-1&0&d&0\\0&0&d&0&d-1&0
		\\d-1&0&0&0&0&d\\d&0&0&0&0&d-1\\0&d-1&0&d&0&0\\0&d&0&d-1&0&0\\\end{array}}.$$
Observe that $M$ is not symmetric.
Hence, \emph{a priori} it is not clear that $M$ is diagonalizable with real eigenvalues.
Nevertheless, a (very tedious) direct computation yields the following.

\begin{lem}\label{Lemma_eigM}
The $6\times 6$ matrix $M$ is diagonalizable and has the non-zero eigenvalues $1$, $2d-1$,
	\begin{equation}\label{eqeigM}
	\Lambda=-\frac{d}2-\frac{\sqrt{d^2-8d+4}}2,\quad\Lambda'=-\frac{d}2+\frac{\sqrt{d^2-8d+4}}2.
	\end{equation}
The eigenspace with eigenvalue $2d-1$ is spanned by $\vecone$.
Moreover, there are two mutually perpendicular eigenvectors $\zeta_2',\zeta_3'$ with eigenvalue $\Lambda$, which satisfy
	$$\|\zeta_2'-(1,1,-1,-1,0,0)^T\|_\infty\leq\frac{10}d,\quad\|\zeta_3'-(1,1,0,0,-1,-1)^T\|_\infty\leq\frac{10}d$$
and $\|\zeta_2'\|=\|\zeta_3'\|$.
\end{lem}
\noindent
Since $M$ describes the operation of $\MM-\KK$ on the subspace $\EE_0$, Lemma~\ref{Lemma_eigM}
implies the following.

\begin{cor}\label{Cor_eigM}
Restricted to the subspace $\EE_0$,
the operator $\MM-\KK$ is diagonalizable with non-zero eigenvalues $1$, $2d-1$, and $\Lambda$, $\Lambda'$ as in~(\ref{eqeigM}).
The vector $e^*=\sum_{i\not=j}e_{ij}$ spans the eigenspace of $2d-1$.
Furthermore, there are two mutually perpendicular eigenvectors $\zeta_2,\zeta_3$ with eigenvalue $\Lambda$, which satisfy
	$$\|\zeta_2-(e_{12}+e_{13}-e_{21}-e_{23})\|_\infty\leq\frac{10}d,\quad\|\zeta_3-(e_{12}+e_{13}-e_{31}-e_{32})\|_\infty\leq\frac{10}d.$$
\end{cor}

Corollary~\ref{Cor_eigM} describes the operation of $\MM-\KK$ on $\EE_0$ completely.
Therefore, as a next step we shall analyze how $\MM-\KK$ operates on $\EE_0^{\perp}$.
More precisely, our goal is to show that restricted to $\EE_0^{\perp}$ the norm of $\MM-\KK$ is significantly smaller than $\Lambda$.
To this end, we observe that the operator $\KK$ merely permutes the coordinates.
Consequently,
	\begin{equation}\label{eqKKbound}
	\|\KK\|\leq1.
	\end{equation}

To to bound the norm of $\MM$ on $\EE_0^{\perp}$, we consider three subspaces of $\EE_0^{\perp}$.
The first subspace $S$ consists of all vectors $\xi\in \EE_0^{\perp}$ such that the value $\xi_{v\ra w}$
only depends on the ``start vertex'' $v$; in symbols,
	$$S=\{\xi\in \EE_0^{\perp}:\forall v\ra w,v\ra u\in\AAA:\xi_{v\ra w}=\xi_{v\ra u}\}.$$
If $\xi\in S$ and $v\in V$, then we let $\xi_{v\ra}=\xi_{v\ra w}$ for any $w\in N(v)$, i.e., $\xi_{v\ra}$ is the ``outgoing value'' of $v$.

The second subspace $T$ consists of all $\xi\in \EE_0^{\perp}$ such that $\xi_{u\ra v}$ depends
only on the ``target vertex'' $v$, i.e.,
	$$T=\{\xi\in \EE_0^{\perp}:\forall u\ra v,w\ra v\in\AAA:\xi_{u\ra v}=\xi_{w\ra v}\}.$$
For $\xi\in T$ and $v\in V$ we let $\xi_{\ra v}=\xi_{u\ra v}$ for any $u\in N(v)$, i.e., $\xi_{\ra v}$ signifies the ``incoming value'' of $v$.

Furthermore, the third subspace $U$ consists of all $\xi$ such that
for any vertex the sum of the ``incoming'' values equals $0$:
	$$U=\cbc{\xi\in \EE_0^{\perp}:\forall v\in V:\sum_{u\in N(v)}\xi_{u\ra v}=0}.$$

\begin{lem}\label{Lemma_STU}
\begin{enumerate}
\item We have $U=\mbox{Kern}(\MM)\cap \EE_0^{\perp}$.
\item Moreover, if $\xi\in T$, then $(\MM\xi)_{v\ra w}=2d\xi_{\ra v}$ for all $v\ra w\in\AAA$.
	In particular, $\MM\xi\in S$.
\item Furthermore, $T\perp U$, and $\EE_0^{\perp}=T\oplus U$.
\end{enumerate}
\end{lem}
\begin{proof}
The first assertion follows immediately from the definition~(\ref{eqDefMMKK}) of $\MM$.
Moreover, if $\xi\in T$, then
	$(\MM\xi)_{v\ra w}=\sum_{u\in N(v)}\xi_{u\ra v}=|N(v)|\xi_{\ra v}=2d\xi_{\ra v}$ due to Lemma~\ref{Lemma_reg},
whence 2.\ follows.
Consequently, if  $\xi\in T$ and $\eta\in U$, then
	$$\scal{\xi}{\eta}=\sum_{u\ra v\in\AAA}\xi_{u\ra v}\eta_{u\ra v}=
		2d\sum_{v\in V}\xi_{\ra v}\sum_{u\in N(v)}\eta_{u\ra v}=0,$$
whence $T\perp U$.
Furthermore, for any $\gamma\in \EE_0^{\perp}$ the vector $\eta$ with entries
	$$\eta_{v\ra w}=\frac1{2d}\sum_{u\in N(w)}\xi_{u\ra w}$$
lies in $T$, because the sum on the r.h.s.\ is independent of $v$.
In addition, $\xi=\gamma-\eta$ satisfies
	$$\sum_{u\in N(v)}\xi_{u\ra v}=\brk{\sum_{u\in N(v)}\gamma_{u\ra v}}-2d\eta_{\ra v}=0\qquad\mbox{for any }v\in V,$$
so that $\xi\in U$.
Hence, any $\gamma\in \EE_0^{\perp}$ can be written as $\gamma=\eta+\xi$ with $\eta\in T$ and $\xi\in U$,
i.e., $\EE_0^{\perp}=T\oplus U$.
\end{proof}

\noindent
By now we have all the prerequisites to analyze the operation of $\MM$ on $\EE_0^\perp$.

\begin{lem}\label{Lemma_boundEperp}
If $\xi\in \EE_0^{\perp}$, then $\|\MM^2\xi\|\leq0.01d^2\|\xi\|$.
\end{lem}
\begin{proof}
Let $\xi\in \EE_0^\perp$.
By the third part of Lemma~\ref{Lemma_STU} there is a decomposition $\xi=\xi_T+\xi_U$ such that $\xi_T\in T$ and $\xi_U\in U$.
Furthermore, the first part of Lemma~\ref{Lemma_STU} entails that $\MM\xi=\MM\xi_T$.
Therefore, we may assume without loss of generality that $\xi=\xi_T\in T$.
Hence, the second part of of Lemma~\ref{Lemma_STU} implies that
	\begin{equation}\label{eqboundEperp1}
	\|\xi'\|=2d\|\xi\|
	\end{equation}
and $\xi'=\MM\xi\in S$.
Consequently, letting $\xi''=\MM\xi'=\MM^2\xi$, we obtain
	\begin{equation}\label{eqboundEperp2}
	\xi''_{v\ra w}=\sum_{u\in N(v)}\xi'_{u\ra v}=\sum_{u\in N(v)}\xi'_{u\ra}.
	\end{equation}
Since the r.h.s.\ of~(\ref{eqboundEperp2}) is independent of $w$, we conclude $\xi''\in S$.

In order to bound $\|\xi''\|=\|\MM^2\xi\|$, we shall express the sum on the r.h.s.\ of~(\ref{eqboundEperp2}) in terms of the adjacency matrix $A(G)$.
To this end, consider the two vectors
	\begin{eqnarray*}
	\eta'=(\eta'_v)_{v\in V}\in\RR^V&\mbox{ with }&\eta'_v=\xi'_{v\ra},\\
	\eta''=(\eta''_v)_{v\in V}\in\RR^V&\mbox{ with }&\eta''_v=\xi''_{v\ra}
	\end{eqnarray*}
for all $v\in V$.
Then
	\begin{eqnarray}\label{eqboundEperp3}
	\|\xi'\|^2&=&\sum_{v\ra w\in\AAA}{\xi'_{v\ra w}}^2=2d\sum_{v\in V}{\xi'_{v\ra}}^2=2d\|\eta'\|^2,\quad\mbox{and analogously }\\
	\|\xi''\|^2&=&2d\|\eta''\|^2.\label{eqboundEperp4}
	\end{eqnarray}
Furthermore, (\ref{eqboundEperp2}) implies that $\eta_v''=\sum_{u\in N(v)}\eta_u'$ for all $v\in V$, i.e.,
	\begin{equation}\label{eqboundEperp5}
	\eta''=A(G)\eta'.
	\end{equation}
Combining~(\ref{eqboundEperp1}), (\ref{eqboundEperp3}), (\ref{eqboundEperp4}), and~(\ref{eqboundEperp5}), we obtain
	\begin{equation}\label{eqboundEperp6}
	\|\MM^2\xi\|=\|\xi''\|=\frac{2d\|A(G)\eta'\|}{\|\eta'\|}\cdot\|\xi\|.
	\end{equation}

Hence, we finally need to bound $\|A(G)\eta'\|$.
To this end, we shall employ our assumption that $G$ is $(d,0.01)$-regular;
namely, condition {\bf R2} from the definition of $(d,0.01)$-regularity entails that
 $\|A(G)\zeta\|\leq0.001d\|\zeta\|$ for all $\zeta\perp\vecone_{V_1},\vecone_{V_2},\vecone_{V_3}$.
Thus, we need to show that $\eta'\perp\vecone_{V_i}$ for $i=1,2,3$.
Assuming w.l.o.g.\ that $i=1$, we have
	\begin{eqnarray}\nonumber
	\scal{\eta'}{\vecone_{V_1}}&=&\sum_{v\in V_1}\xi'_{v\ra}
		=(2d)^{-1}\sum_{v\ra w\in\AAA:v\in V_1}\xi'_{v\ra w}=(2d)^{-1}\scal{\xi'}{e_{12}+e_{13}}\\
		&=&(2d)^{-1}\scal{\MM\xi}{e_{12}+e_{13}}=(2d)^{-1}\scal{\xi}{\MM^T(e_{12}+e_{13})}.
		\label{eqboundEperp7}
	\end{eqnarray}
Further, as $\MM^T(e_{12}+e_{13})\in\EE_0$ by Lemma~\ref{Lemma_MMT}, while $\xi\in \EE_0^{\perp}$ by our assumption,
(\ref{eqboundEperp7}) implies that $\scal{\eta'}{\vecone_{V_1}}=0$.
Consequently, we obtain that $\|A(G)\eta'\|\leq0.001d\|\eta'\|$, whence (\ref{eqboundEperp6}) yields the assertion.
\end{proof}

\medskip\noindent
\emph{Proof of Proposition~\ref{Prop_dominant}.}
Combining Corollary~\ref{Cor_eigM} with the tensor product representation~(\ref{eqStartgame2}) of $\LL$,
we conclude that the six vectors
	\begin{equation}\label{eqzetatensors}
		\zeta_j^1=\bc{\begin{array}{c}1\\0\\0\end{array}}\tensor\zeta_j,\
		\zeta_j^2=\bc{\begin{array}{c}0\\1\\0\end{array}}\tensor\zeta_j,\	
		\zeta_j^3=\bc{\begin{array}{c}0\\0\\1\end{array}}\tensor\zeta_j\qquad(j=2,3)
	\end{equation}
are eigenvectors of $\LL$ with eigenvalue $\lambda=-\frac12\Lambda$.
In addition, the tensor representation~(\ref{eqzetatensors}) of the vectors $\zeta_j^a$  immediately
implies the symmetry statement~(\ref{eqzetasymm}), while~(\ref{eqzetanorm}) follows from Corollary~\ref{Cor_eigM}.
Moreover, once more by Corollary~\ref{Cor_eigM} the three vectors
	$$e^1=\bc{\begin{array}{c}1\\0\\0\end{array}}\tensor e^*,\
		e^2=\bc{\begin{array}{c}0\\1\\0\end{array}}\tensor e^*,\
		e^3=\bc{\begin{array}{c}0\\0\\1\end{array}}\tensor e^*$$
are eigenvalues with eigenvector $-\frac12(2d-1)=\frac12-d$,
and all other eigenvalues of $\LL$ restricted to $\EE$ are $\leq\frac12$ in absolute value.
In addition, Lemma~\ref{Lemma_MMT} shows in combination with~(\ref{eqStartgame2}) that $\LL\EE,\LL^T\EE\subset\EE$.
Finally, Lemma~\ref{Lemma_boundEperp} implies in combination with~(\ref{eqStartgame2})
that $\|\LL^2\xi\|\leq0.01d^2\|\xi\|$ for all $\xi\perp\EE$.
\qed


\subsection{Proof of Proposition~\ref{Prop_initial}}\label{Sec_initial}

Before we get to the proof, let us briefly discuss why the assertion (i.e., Proposition~\ref{Prop_initial}) is plausible.
In fact, let us point out that the vector $\Delta(0)$ is easily seen to satisfy {\bf F2} with probability $\Omega(1)$.
For each of the inner products $\scal{\Delta(0)}{\zeta_i^a}$ is a sum of $n$ independent random variables, whence the
central limit theorem implies that
	$$\sqrt{n}\|\Delta(0)\|^{-1}\|\zeta_i^a\|^{-1}\scal{\Delta(0)}{\zeta_i^a}$$
is asymptotically normal
(the factor $\sqrt{n}\|\Delta(0)\|^{-1}\|\zeta_i^a\|^{-1}$, which is independent of the random vector $\Delta(0)$,
is needed to ensure that mean and variance are of order $\Theta(1)$).
In fact, since the vectors $(\zeta_i^a)_{a=1,2,3;\,i=2,3}$ are mutually perpendicular, the \emph{joint} distribution
of the random variables
	$$(\sqrt{n}\|\Delta(0)\|^{-1}\|\zeta_i^a\|^{-1}\scal{\Delta(0)}{\zeta_i^a})_{i=2,3;a=1,2,3}$$
is asymptotically a (multivariate) Gaussian.
Therefore, the probability that $\Delta(0)$ satisfies {\bf F2} is $\Omega(0)$.

However, once we condition on $\Delta(0)$ satisfying {\bf F1},
the entries of $\Delta(0)$ are not independent anymore, whence the above argument does not
yield a bound on the probability that $\Delta(0)$ satisfies both {\bf F1} and {\bf F2}.
Nonetheless, the dependence of the entries of $\Delta(0)$ is weak enough
to allow for an elementary direct analysis.
We begin with bounding the probability that $\Delta(0)$ satisfies {\bf F1}.
To this end, we define a partition $(W_1,W_2,W_3)$ of $V$ by letting
	$$W_i=\{v\in V:\Delta_{v\ra w}^i=\delta\mbox{ for all }w\in N(v)\};$$
in other words, $W_i$ consists of all vertices for which the random number $a$ chosen in Step~1 of \texttt{BPCol} was equal to $i$.

\begin{lem}\label{Lemma_F1}
The probability that $\Delta(0)$ satisfies {\bf F1} is $\Omega(n^{-1})$.
\end{lem}
\begin{proof}
A sufficient condition for $\Delta(0)$ to satisfy {\bf F1} is that $W_1=W_2=W_3=\frac{n}3$.
Moreover, the total number of vectors that can be generated by Step~1 of \texttt{BPCol} equals $3^n$,
out of which $\bink{n}{n/3\,n/3\,n/3}$ yield $W_1=W_2=W_3=\frac{n}3$.
Therefore, the assertion follows from Stirling's formula.
\end{proof}

\emph{In the remainder of this section we condition on the event that $\Delta(0)$ is such that $W_1=W_2=W_3$.}
Thus, $(W_1,W_2,W_3)$, is just a random partition of $V$ into three classes of equal size,
and for all $v\in W_i$, all $j\in\{1,2,3\}\setminus\{i\}$, and all $w\in N(v)$ we have
	$$\Delta_{v\ra w}^i=\delta,\quad\Delta_{v\ra w}^j=-\frac{\delta}2.$$

\begin{lem}\label{Lemma_multis}
For any constant $c_1>0$ there exists a constant $c_2>0$ such that the following holds.
If $(s_i^a)_{i,a=1,2,3}$ are integers of absolute value $|s_i^a|\leq c_1\sqrt{n}$ such
that $\sum_{a=1}^3s_j^a=\sum_{i=1}^3s_i^b=0$ for all $1\leq b,j\leq 3$, then
	$$\pr\brk{\forall 1\leq a,i\leq 3:|V_a\cap W_i|=\frac{n}9+s_i^a}\geq c_2n^{-2}.$$
\end{lem}
\begin{proof}
The sets $W_1,W_2,W_3$ are randomly chosen mutually disjoint subsets of $V$ of cardinality $n/3$ each,
whereas $V_1,V_2,V_3$ are fixed subsets of $V$.
Therefore, the total number of ways to choose $W_1,W_2,W_3$ is given by the multinomial coefficient
	$\bink{n}{n/3,n/3,n/3}$;
by Stirling's formula,
	\begin{equation}\label{eqmultis1}
	\bink{n}{n/3,n/3,n/3}\leq 10 n^{-1} 3^{n}.
	\end{equation}
Moreover, the number of ways to choose $W_1,W_2,W_3$ such that $|V_a\cap W_i|=s_i^a$ equals
	\begin{equation}\label{eqmultis2}
	\prod_{a=1}^3\bink{n/3}{n/9+s_1^a,n/9+s_2^a,n/9+s_3^a}
	\end{equation}
(because the $a$'th factor on the r.h.s.\ equals the number of ways to partition $V_a$ into three pieces $V_a\cap W_1$, $V_a\cap W_2$, $V_a\cap W_3$
of the desired sizes).
Combining~(\ref{eqmultis1}) and~(\ref{eqmultis2}) with Stirling's formula, we get
	\begin{eqnarray}\nonumber
	\pr\brk{\forall 1\leq a,i\leq 3:|V_a\cap W_i|=\frac{n}9+s_i^a}&\geq&\frac{n(n/3)!^3}{10\cdot3^n\prod_{1\leq i,a\leq3}(n/9+s_i^a)!}\\
		&\geq&\frac{n^{5/2+n}}{10\cdot (9\mathrm{e})^n\prod_{1\leq i,a\leq3}(n/9+s_i^a)!}.
		\label{eqmultis3}
	\end{eqnarray}
Furthermore, once more due to Stirling's formula,
	\begin{eqnarray}
	(n/9+s_i^a)!&\leq&\exp(-n/9-s_i^a)(n/9+s_i^a)^{n/9+s_i^a}\sqrt{n}\nonumber\\
		&=&\exp(-n/9-s_i^a)(n/9)^{n/9+s_i^a}(1+9s_i^a/n)^{n/9+s_i^a}\sqrt{n}\nonumber\\
		&\leq&\exp(-n/9+9{s_i^a}^2/n)(n/9)^{n/9+s_i^a}.
		\label{eqmultis4}
	\end{eqnarray}
Since we are assuming that $s_i^a\leq c_1\sqrt{n}$ and $\sum_{i=1}^3s_i^a=0$, (\ref{eqmultis4}) entails that
	\begin{equation}\label{eqmultis5}
	\prod_{1\leq i,a\leq 3}(n/9+s_i^a)!\leq (n/9\mathrm{e})^n n^{9/2}\exp(9\sum_{a,i}{s_i^a}^2/n)\leq c_2'(n/9\mathrm{e})^n n^{9/2}
	\end{equation}
for a bounded number $c_2'$ that depends only on $c_1$.
Finally, plugging~(\ref{eqmultis5}) into~(\ref{eqmultis3}) and cancelling, we obtain the assertion.
\end{proof}

\begin{cor}\label{Cor_multis}
For any two constants $c_3,\beta>0$ there exists a constant $c_4>0$ such that the following holds.
If $(t_i^a)_{i,a=1,2,3}$ are numbers of absolute value $|t_i^a|\leq c_3$ such
that $\sum_{a=1}^3t_j^a=\sum_{i=1}^3t_i^b=0$ for all $1\leq b,j\leq 3$, then
	$$\pr\brk{\forall 1\leq a,i\leq 3:|n^{-\frac12}(|V_a\cap W_i|-\frac{n}9)-t_i^a|\leq\beta}\geq c_4.$$
\end{cor}
\begin{proof}
Let $S$ be the set of all tuples $(s_i^a)_{a,i=1,2,3}$ of integers such that
$|n^{-\frac12}s_j^b-t_j^b|\leq\beta$, and
$\sum_{a=1}^3s_j^a=\sum_{i=1}^3s_i^b=0$ for all $1\leq b,j\leq 3$.
Then $|S|\geq\beta^4 n^2/32$.
Moreover, all $(s_i^a)_{a,i=1,2,3}\in S$ satisfy $|s_j^b|\leq (c_3+1)\sqrt{n}$ ($1\leq b,j\leq 3$).
Therefore, Lemma~\ref{Lemma_multis} (applied with $c_1=c_3+1$) shows that
	\begin{eqnarray*}
	\pr\brk{\forall a,i:|n^{-\frac12}(|V_a\cap W_i|-\frac{n}9)-t_i^a|\leq\beta}
		&\geq&\sum_{(s_i^a)\in S}
			\pr\brk{\forall a,i:|V_a\cap W_i|=\frac{n}9+s_i^a}\\
		&\geq&c_2|S|n^{-2}\geq\beta^4c_2/32,
	\end{eqnarray*}
as desired.
\end{proof}

Since the vector $\Delta(0)$ just represents the partition $W_1,W_2,W_3$, and the vectors
$\sum_{j\not=i}e_{ij}^a$ just represents the coloring $V_1,V_2,V_3$, Corollary~\ref{Cor_multis} easily implies
a result on the joint distribution of the inner products $\scal{\Delta(0)}{\sum_{j\not=i}e_{ij}^a}$.

\begin{cor}\label{Cor_innerproducts}
For any two constants $c_5,\gamma>0$ there exists a constant $c_6>0$ such that the following is true.
Suppose that $(z_i^a)_{1\leq a,i\leq 3}$ are numbers such that $|z_j^b|\leq c_5$
and $\sum_{i=1}^3z_i^b=\sum_{a=1}^3z_j^a=0$ for all $1\leq b,j\leq 3$.
Then
	$$\pr\brk{\forall a,i:|z_i^a-\frac{\scal{\Delta(0)}{\sum_{j\not=i}e_{ij}^a}}{\|\Delta(0)\|\|\zeta_i^a\|}\cdot\sqrt{n}|\leq\gamma}
		\geq c_6.$$
\end{cor}
\begin{proof}
The definition of $\eta(0)$ in Step~1 of \texttt{BPCol} shows that
	\begin{equation}\label{eqinnerproducts0}
	\Delta_{v\ra w}^a(0)=\eta_{v\ra w}^a(0)-\frac13=\left\{\begin{array}{cl}
		\delta&\mbox{ if }v\in W_a,\\
		-\delta/2&\mbox{ otherwise.}\end{array}
		\right.\qquad\mbox{for all }v\ra w\in\AAA.
	\end{equation}
Therefore,
	\begin{equation}\label{eqinnerproducts1}
	\|\Delta(0)\|=\sqrt{3dn/2}\cdot\delta.
	\end{equation}
Moreover, by Proposition~\ref{Prop_dominant} there is a number $0.99\leq c_7\leq1.01$ such that
	\begin{equation}\label{eqinnerproducts2}
	\|\zeta_i^a\|=c_7\|e_{12}^a+e_{13}^a-e_{21}^a-e_{23}^a\|=2c_7\sqrt{dn}.
	\end{equation}
Furthermore, using~(\ref{eqinnerproducts0}),
we can easily compute the scalar product $\scal{\Delta(0)}{\sum_{j\not=i}e_{ij}^a}$ ($1\leq a,i\leq 3$):
	\begin{eqnarray}\nonumber
	\scal{\Delta(0)}{\sum_{j\not=i}e_{ij}^a}&=&
		\sum_{v\ra w\in\AAA:v\in V_i}\Delta_{v\ra w}^a(0)
		=|V_i\cap W_a|\cdot d\delta-|V_i\setminus W_a|\cdot\frac{d\delta}2\\
		&=&\frac{3d\delta}2(|V_i\cap W_a|-n/9)\qquad\mbox{[because $|V_i|=|W_a|=n/3$]}.
		\label{eqinnerproducts3}
	\end{eqnarray}
Combining~(\ref{eqinnerproducts1}), (\ref{eqinnerproducts2}), and~(\ref{eqinnerproducts3}), we conclude that
for a certain constant $c_8>0$
	$$
	\frac{\scal{\Delta(0)}{\sum_{j\not=i}e_{ij}^a}}{\|\Delta(0)\|\|\zeta_i^a\|}\cdot\sqrt{n}
		=\frac{c_8}{\sqrt{n}}\cdot(|V_i\cap W_a|-n/9).
	$$
Therefore, the assertion follows from Corollary~\ref{Cor_multis} by setting $s_i^a=c_4^{-1}\sqrt{n}\cdot z_i^a$
and $\beta=\gamma/c_8$.
\end{proof}

\medskip\noindent
\emph{Proof of Proposition~\ref{Prop_initial}.}
Let $\alpha=\exp(-1/\epsilon)$ and
	\begin{equation}\label{eqQminus0}
	\hat x_i^a=\left\{\begin{array}{cl}
		-1&\mbox{ if }a=i,\\
		1/2&\mbox{ otherwise}
		\end{array}\right.\qquad(i=2,3;\,a=1,2,3).
	\end{equation}
Then the definitions~(\ref{eqxia}) and~(\ref{eqyia}) of the variables $x_i^a$ and $y_i^a$ entail that
	\begin{eqnarray}\nonumber
	\pr\brk{\forall a,i\in\{1,2,3\},i\not=a:|y_a^a-1|<\alpha\wedge|y_i^a-1/2|<\alpha}\\
		&\hspace{-6cm}\geq&\hspace{-3cm}\pr\brk{\forall a,i\in\{1,2,3\}:|x_i^a-\hat x_i^a|<\alpha/2}.
	\label{eqinitiallower}
	\end{eqnarray}
Therefore, we shall derive a lower bound on $\pr\brk{\forall a,i:|x_i^a-\hat x_i^a|<\alpha/2}$.

To this end, let
	$$e_i^a=\sum_{j\in\{1,2,3\}\setminus\{i\}}e_{ij}^a\qquad(1\leq a,i\leq 3),$$
and let $\mathcal{V}\subset\RR^{\AAA}$ be the space spanned by these nine vectors.
In addition, let $q:\RR^{\AAA}\ra\mathcal{V}$ be the orthogonal projection onto $\mathcal{V}$.
Since the construction of the initial vector $\Delta(0)$ in Step~1 of \texttt{BPCol} ensures
that $\Delta(0)\in\mathcal{V}$, we have
	\begin{eqnarray*}
	\frac{\|\Delta(0)\|\cdot\|\zeta_i^a\|}{\sqrt{n}}\cdot x_i^a&=&\scal{\Delta(0)}{\zeta_i^a}
		=\scal{q\Delta(0)}{\zeta_i^a}=\scal{\Delta(0)}{q\zeta_i^a}.
	\end{eqnarray*}
Hence, instead of the vectors $\zeta_i^a$ we may work with their projections $q\zeta_i^a$ onto $\mathcal{V}$.
Thus, let $q_{ij}^a\in\RR$ be the coefficients such that
	$$q\zeta_i^a=\sum_{j=1}^3q_{ij}^ae_j^a\qquad\mbox{($i=2,3$, $a=1,2,3$)}.$$
Then by symmetry we have $q_{ij}^a=q_{ij}^b$ for all $1\leq a,b\leq 3$;
therefore, we will briefly write $q_{ij}$ instead of $q_{ij}^a$.
Furthermore, (\ref{eqzeta}) implies the bounds
	\begin{eqnarray}&\label{eqsmallq1}
		0.99\leq q_{21}\leq1.01,\
		-1.01\leq q_{22}\leq-0.99,\
		-0.01\leq q_{23}\leq0.01,\\&
	0.99\leq q_{31}\leq1.01,\
		-0.01\leq q_{32}\leq-0.01,\
		-1.01\leq q_{33}\leq-0.99.
		\label{eqsmallq2}
	\end{eqnarray}
As a consequence, the matrix
	$$Q=\bc{\begin{array}{ccc}
		q_{21}&q_{22}&q_{23}\\
		q_{31}&q_{32}&q_{33}\\
		1&1&1
		\end{array}}$$
is regular, and there is a constant $c_9>0$ such that $\|Q^{-1}\|\leq c_9$.

Let
	\begin{equation}\label{eqQminus1}
		\bc{\begin{array}{c}z_1^a\\z_2^a\\z_3^a\end{array}}=
		Q^{-1}\bc{\begin{array}{c}\hat x_2^a\\\hat x_3^a\\0\end{array}}\qquad(a=1,2,3).
	\end{equation}
Since $\|Q^{-1}\|\leq c_9$ and $|x_i^a|\leq1$ for all $a,i$, we have
	\begin{equation}\label{eqQminus2}
	|z_i^a|\leq 5c_9\qquad\mbox{($1\leq a,i\leq 3$)}.
	\end{equation}
In addition, (\ref{eqQminus0}) and~(\ref{eqQminus1}) imply that
	\begin{eqnarray}\label{eqQminus3}
	\sum_{a=1}^3\bc{\begin{array}{c}z_1^a\\z_2^a\\z_3^a\end{array}}&=&
		Q^{-1}\brk{\sum_{a=1}^3\bc{\begin{array}{c}x_2^a\\x_3^a\\0\end{array}}}=0,\mbox{ and}\\
	\sum_{i=1}^3z_i^b&=&0\qquad(1\leq b\leq 3).	\label{eqQminus4}
	\end{eqnarray}
Combining~(\ref{eqQminus2})--(\ref{eqQminus4}), we see that
$(z_i^a)_{1\leq a,i\leq 3}$ satisfies the assumptions of Corollary~\ref{Cor_innerproducts}, whence
	\begin{equation}\label{eqQminus5}
	\pr\brk{\forall a,i:|z_i^a-\frac{\scal{\Delta(0)}{e_i^a}}{\|\Delta(0)\|\|\zeta_i^a\|}\cdot\sqrt{n}|\leq\alpha^2}
		\geq c_6
	\end{equation}
for some constant $c_2>0$.
Furthermore, if $\Delta(0)\in\RR^{\AAA}$ satisfies
$|z_i^a-\frac{\scal{\Delta(0)}{e_i^a}}{\|\Delta(0)\|\|\zeta_i^a\|}\cdot\sqrt{n}|\leq\alpha^2$,
then~(\ref{eqQminus1}) and the bounds~(\ref{eqsmallq1})--(\ref{eqsmallq2}) imply that
	\begin{eqnarray*}
	|\hat x_j^a-x_j^a|&=&|\hat x_j^a-\frac{\scal{\Delta(0)}{\zeta_j^a}}{\|\Delta(0)\|\|\zeta_j^a\|}\cdot\sqrt{n}|
		=|\sum_{i=1}^3q_{ji}\bc{z_i^a-\frac{\scal{\Delta(0)}{e_i^a}}{\|\Delta(0)\|\|\zeta_j^a\|}\cdot\sqrt{n}}|\\
		&\leq&\alpha^2\sum_{i=1}^3|q_{ji}|\leq3\alpha^2<\alpha/2\qquad(j=2,3;\,a=1,2,3).
	\end{eqnarray*}
Therefore, (\ref{eqQminus5}) yields
	$\pr\brk{\forall a,i:|x_i^a-\hat x_i^a|<\alpha/2}\geq c_6$.
Thus, the assertion follows from~(\ref{eqinitiallower}) and Lemma~\ref{Lemma_F1}.
\qed

\subsection{Proof of Proposition~\ref{Prop_err}}\label{Sec_err}

Our goal in this section is to bound the error $\|\Delta(l)-\Xi(l)\|_\infty$ resulting from replacing the non-linear operator
$\BB$ by the linear operator $\LL$.
Since $\Delta(l)=\BB^l\Delta(0)$ and $\Xi(l)=\LL^l\Xi(0)=\LL^l\Delta(0)$ by~(\ref{eqXiLL}), the main difficulty of this analysis is
to bound how errors that were made early on in the sequence (i.e., for ``small'' $l$) amplify in the subsequent iterations.
To control this phenomenon, we shall proceed by induction on $l$.
We begin with a simple lemma that bounds the error occurring in a single iteration.
Recall that the constructions of $\Xi(l)$ and $\Delta(l)$ ensure that
$\sum_{a=1}^3\Xi_{v\ra w}^a(l)=\sum_{a=1}^3\Delta_{v\ra w}^a(l)=0$ for all $v\ra w\in\AAA$
(cf.~(\ref{eqIniI}) and~(\ref{eqXiLLsum})).

\begin{lem}\label{Lemma_L}
Suppose that $\Gamma$ satisfies $\sum_{a=1}^3\Gamma_{v\ra w}^a=0$ for all $v\ra w\in\AAA$.
If $\|\Gamma\|_{\infty}<0.001d^{-1}$, then $\|\BB\Gamma-\LL\Gamma\|_\infty\leq 100d^2\|\Gamma\|_\infty^2$.
\end{lem}
\begin{proof}
We employ the elementary inequalities
	\begin{equation}\label{eqL1}
	\exp(-x-x^2)\leq1-x\leq\exp(-x)\leq1-x+x^2\qquad(|x|\leq0.1).
	\end{equation}
Let $v\ra w\in\AAA$, $a\in\{1,2,3\}$, and set
	$$\Pi_b=\prod_{u\in N(v)\setminus w}1-\frac32\Gamma_{u\ra v}^b\qquad(b\in\{1,2,3\}).$$
Moreover, let $\hat\Gamma=\BB\Gamma$.
Then we can rephrase the definition~(\ref{eqBPOpDelta}) of $\BB$ as
	\begin{equation}\label{eqL2}
	\hat\Gamma_{v\ra w} ^a=-\frac13+\frac{\Pi_a}{\sum_{b=1}^3\Pi_b}.
	\end{equation}

In order to prove the lemma, we shall bound the error term
	$|\Pi_b-(1-\frac32\sum_{u\in N(v)\setminus w}\Gamma_{u\ra v}^b)|$.
To this end, note that by~(\ref{eqL1}) there exist numbers
	$0\leq\alpha_{u\ra v}^b\leq9/4$
such that
	$1-\frac32\Gamma_{u\ra v}^b=\exp(-\frac32\Gamma_{u\ra v}^b-\alpha_{u\ra v}^b\Gamma_{u\ra v}^{b\,2}).$
Hence, once more by~(\ref{eqL1}) there is a number $-1\leq\beta^b\leq1$ such that
	\begin{eqnarray}
	\Pi_b&=&\exp\brk{-\sum_{u\in N(v)\setminus w}\frac32\Gamma_{u\ra v}^b+\alpha_{u\ra v}^b\Gamma_{u\ra v}^{b\,2}}\nonumber\\
		&=&1-\sum_{u\in N(v)}\brk{\frac32\Gamma_{u\ra v}^b+\alpha_{u\ra v}^b\Gamma_{u\ra v}^{b\,2}}
			+\beta^b\brk{\sum_{u\in N(v)}\frac32\Gamma_{u\ra v}^b+\alpha_{u\ra v}^b\Gamma_{u\ra v}^{b\,2}}^2\nonumber\\
		&=&L_b+E_b,\qquad\qquad\mbox{  \ where we let}\label{eqL4}\\
	L_b&=&1-\sum_{u\in N(v)}\frac32\Gamma_{u\ra v}^b,\mbox{ and}\nonumber\\
	E_b&=&\sum_{u\in N(v)\setminus w}\alpha_{u\ra v}^b\Gamma_{u\ra v}^{b\,2}+
			\beta^b\brk{\sum_{u\in N(v)}\frac32\Gamma_{u\ra v}^b+\alpha_{u\ra v}^b\Gamma_{u\ra v}^{b\,2}}^2.
			\nonumber
	\end{eqnarray}
Further, since $\|\Gamma\|_\infty\leq0.001/d$ by assumption and $|N(v)|=2d$ by Lemma~\ref{Lemma_reg}, we obtain the bound
	\begin{equation}\label{eqL3}
	|E_b|\leq10d^2\|\Gamma\|_{\infty}^2\leq0.01.
	\end{equation}

As $\sum_{b=1}^3L_b=3$ due to our assumption that $\sum_{b=1}^3\Gamma_{u\ra v}^b=0$ for all $u\ra v\in\AAA$,
plugging~(\ref{eqL4}) into~(\ref{eqL2}) yields
	\begin{eqnarray}\label{eqL4a}
	\hat\Gamma_{v\ra w}^a+\frac13&=&\frac{L_a+E_a}{3+E_1+E_2+E_3}
		=\frac{L_a}3+\frac{3E_a+L_a(E_1+E_2+E_3)}{3(3+E_1+E_2+E_3)}.
	\end{eqnarray}
Since $|L_a|\leq1+3d\|\Gamma\|_\infty\leq2$, (\ref{eqL3}) and~(\ref{eqL4a}) yield that
	\begin{equation}\label{eqL5}
	|\frac13(1-L_a)-\hat\Gamma_{v\ra w}^a|\leq100d^2\|\Gamma\|_{\infty}^2.
	\end{equation}
Finally, a glance at~(\ref{eqDefL}) reveals that $(\LL\Gamma)_{v\ra w}^a=\frac13(1-L_a)$,
and thus the assertion follows from~(\ref{eqL5}).
\end{proof}

Lemma~\ref{Lemma_L} allows us to bound the error $\|\Delta(l+1)-\Xi(l+1)\|_\infty$
resulting from iteration $l+1$ in terms of the error $\|\Delta(l)-\Xi(l)\|_\infty$ from the previous iteration.
In the sequel we let $C>0$ denote a sufficiently large constant.

\begin{lem}\label{Lemma_err}
Suppose that $\|\Delta(l)-\Xi(l)\|_\infty\leq(Cd)^{-1}$.
Then
	$$\|\Delta(l+1)-\Xi(l+1)\|_\infty\leq 2Cd^2\|\Xi(l)\|_\infty^2+4d\|\Delta(l)-\Xi(l)\|_\infty.$$
\end{lem}	
\begin{proof}
By Lemma~\ref{Lemma_L} and the definition~(\ref{eqDefL}) of $\LL$ we have
	\begin{eqnarray}\nonumber
	\|\Delta(l+1)-\Xi(l+1)\|_\infty&=&\|\BB\Delta(l)-\LL\Xi(l)\|_\infty\\
		&\leq&\|\BB\Delta(l)-\LL\Delta(l)\|_\infty+\|\LL\Delta(l)-\LL\Xi(l)\|_\infty\nonumber\\
		&\leq&Cd^2\|\Delta(l)\|_\infty^2+2d\|\Delta(l)-\Xi(l)\|_\infty.\label{eqErr}
	\end{eqnarray}
Moreover,
	$\|\Delta(l)\|_\infty\leq\|\Xi(l)\|_\infty+\|\Xi(l)-\Delta(l)\|_\infty$, whence (\ref{eqErr}) yields
	\begin{eqnarray*}
	\|\Delta(l+1)-\Xi(l+1)\|_\infty&\leq&2Cd^2\brk{\|\Xi(l)\|_\infty^2+\|\Xi(l)-\Delta(l)\|_\infty^2}+2d\|\Delta(l)-\Xi(l)\|_\infty.
	\end{eqnarray*}
This implies the assertion, because we are assuming that $\|\Delta(l)-\Xi(l)\|_\infty\leq(Cd)^{-1}$.
\end{proof}
\noindent
Further, applying Lemma~\ref{Lemma_err} $L$ times recursively, we obtain the following bound.

\begin{cor}\label{Cor_err}
Suppose that $\|\Delta(l)-\Xi(l)\|_\infty\leq(Cd)^{-1}$ for all $l<L$.
Then
	$$\|\Delta(L)-\Xi(L)\|_\infty\leq2Cd^2\sum_{j=1}^{L-1}(4d)^{j-1}\|\Xi(L-j)\|_\infty^2+Cd^2(4d)^{L-1}\|\Delta(0)\|_\infty^2.$$
\end{cor}

To proceed, we need the following (rough) absolute bound on the error $\|\Delta(L)-\Xi(L)\|_\infty$.

\begin{lem}\label{Lemma_errEarly}
If $L\leq\log^2n$, then
$\|\Delta(L)-\Xi(L)\|_\infty<(Cd)^{-1}$.
\end{lem}
\begin{proof}
The proof is by induction on $l$.
For $L=0$ the assertion is trivially true.
Thus, assume that $\|\Delta(l)-\Xi(l)\|_\infty<(Cd)^{-1}$ for all $l<L\leq\log^2n$.
Then Corollary~\ref{Cor_err} entails that
	$$\|\Delta(L)-\Xi(L)\|_\infty\leq2Cd^2\sum_{j=1}^{L-1}(4d)^{j-1}\|\Xi(L-j)\|_\infty^2+Cd^2(4d)^{L-1}\|\Delta(0)\|_\infty^2.$$
Further, the definition~(\ref{eqDefL}) of $\LL$ shows that
	$$\|\Xi(l)\|_\infty\leq(2d)^l\|\Delta(0)\|_\infty=(2d)^l\delta.$$
Hence,
	$$\|\Delta(L)-\Xi(L)\|_\infty\leq4Cd^2(2d)^{2L-2}\delta^2+Cd^2(4d)^{L-1}\delta^2.$$
As $\delta\leq\exp(-\log^{3}n)$ and $d=O(1)$, the r.h.s.\ is $o(1)$ as $n\ra\infty$,
 and thus $\|\Delta(L)-\Xi(L)\|_\infty<(Cd)^{-1}$, provided that $n$ is sufficiently large.
\end{proof}

\begin{lem}\label{Lemma_errLate}
Let $L^*$ be the maximum integer such that $\|\Xi(L^*)\|_\infty<\eps$.
Then for all $\log^2n\leq L\leq L^*$ we have
$\|\Xi(L)-\Delta(L)\|_\infty\leq -\log(\epsilon)\cdot\|\Xi(L)\|_\infty^2$.
\end{lem}
\begin{proof}
By the definition~(\ref{eqDefL}) of $\LL$ there are constants $c_1,c_2>0$ such that
	\begin{eqnarray}\label{eqerrLate1}
	\|\Xi(l)\|_\infty&\leq&(2d)^l\delta\qquad(\forall\,l\leq c_2\log n),\\
	\|\Xi(l)\|_\infty&\in&\brk{c_1^{-1}\lambda^l\delta/\sqrt{dn},c_1\lambda^l\delta/\sqrt{dn}}\qquad(\forall\,l\geq c_2\log n).\label{eqerrLate2}
	\end{eqnarray}
We proceed inductively for $\log^2n\leq L\leq L^*$.
Thus, assume that $\|\Xi(l)-\Delta(l)\|_\infty\leq c_1\|\Xi(l)\|_\infty^2$ for all $\log^2n\leq l<L$.
Since $\lambda\geq0.1d$ and $\|\Xi(L)\|_\infty<\eps$, this implies that
	$$\|\Xi(l)-\Delta(l)\|_\infty\leq (Cd)^{-1}\quad\mbox{ for all }\log^2n\leq l<L.$$
Furthermore, $\|\Xi(l)-\Delta(l)\|_\infty<(Cd)^{-1}$ for all $l<\log^2n$ by Lemma~\ref{Lemma_errEarly}.
Therefore, we can apply Corollary~\ref{Cor_err} to obtain
	\begin{eqnarray}\label{eqerrLate3}
	\|\Xi(L)-\Delta(L)\|_\infty&\leq&
		2Cd^2\sum_{j=1}^{L-1}(4d)^{j-1}\|\Xi(L-j)\|_\infty^2+Cd^2(4d)^{L-1}\|\Delta(0)\|_\infty^2.
	\end{eqnarray}
Since $L\geq\log^2n$ and $\lambda\geq0.1d$, (\ref{eqerrLate1}) and (\ref{eqerrLate2})
imply that the sum on the r.h.s.\ of~(\ref{eqerrLate3}) is dominated by the term for $j=L-1$.
Hence,
	\begin{eqnarray}\nonumber
	\|\Xi(L)-\Delta(L)\|_\infty&\leq&
		4Cd^2\|\Xi(L-1)\|_\infty^2+Cd^2(4d)^{L-1}\delta^2\\
		&\leq&c_3d^2\delta^2\brk{n^{-1}\lambda^{2L-2}+(4d)^{L-1}}\nonumber\\
		&\leq&2c_3d^2\delta^2\lambda^{2L-2}n^{-1}\leq c_4\delta^2\lambda^{2L}n^{-1}.
	\label{eqerrLate4}
	\end{eqnarray}
Combining~(\ref{eqerrLate2}) and~(\ref{eqerrLate4}), we conclude that $\|\Xi(L)-\Delta(L)\|_\infty<-\log(\epsilon)\cdot\|\Xi(L)\|_\infty^2$
(provided that $\epsilon$ is chosen small enough).
\end{proof}
\noindent
Finally, Proposition~\ref{Prop_err} follows from Lemma~\ref{Lemma_errLate} directly.

\subsection{Proof of Proposition~\ref{Prop_Endgame}}\label{Sec_Endgame}

Let $\mu=\nu\lambda^{L_2}$.
Then Corollary~\ref{Cor_Startgame} and Proposition~\ref{Prop_initial} entail that
	\begin{eqnarray}\label{eq_EndgameI}
	(1-\eps^3)\mu&\leq&\Delta_{v\ra w}^a(L_2)\leq(1+\eps^3)\mu\qquad\mbox{if $v\in V_a$ and $w\in N(v)$, and}\\
	(-\frac12-\eps^3)\mu&\leq&\Delta_{v\ra w}^a(L_2)\leq(-\frac12+\eps^3)\mu\qquad\mbox{if $v\not\in V_a$ and $w\in N(v)$.}
		\label{eq_EndgameII}
	\end{eqnarray}
To prove Proposition~\ref{Prop_Endgame}, we consider two cases.
The first case is that $\|\Delta(L_2)\|_\infty\leq(\eps d)^{-1}$ is ``small''.
Then it will take two more steps for the messages to properly represent the coloring $(V_1,V_2,V_3)$, i.e., $L_3=L_2+2$.
By contrast, if $\|\Delta(L_2)\|_\infty>(\eps d)^{-1}$ is ``large'', we will just need one more step ($L_3=L_2+1$).
In both cases the proof is based on a direct analysis of the BP equations~(\ref{eqBPOp}).

\begin{lem}\label{Lemma_EndgameMedium}
If $0.01\eps d^{-1}\leq\|\Delta(L_2)\|_\infty\leq(\eps d)^{-1}$, then
\begin{eqnarray}
	\eta_{u\ra v}^i(L_2+1)&=&\left\{\begin{array}{cl}
		\frac13+(1+\gamma(u,v,i))\beta&\mbox{ if }u\in V_i,\\
		\frac13-(1+\gamma(u,v,i))\beta'&\mbox{ otherwise,}
		\end{array}\right.
	\end{eqnarray}
where $|\gamma(u,v,i)|\leq\eps^3$ and $\beta,\beta'>\eps^2$.
\end{lem}
\begin{proof}
We have
	\begin{eqnarray}
	\eta_{v\ra w}^i(L_2+1)&=&
		\frac{\prod_{u\in N(v)\setminus w}1-\frac32\Delta_{u\ra v}^i(L_2)}%
			{\sum_{j=1}^3\prod_{u\in N(v)\setminus w}1-\frac32\Delta_{u\ra v}^j(L_2)}\nonumber\\
			&=&\frac{\exp\bc{-\frac32\sum_{u\in N(v)\setminus w}\Delta_{u\ra v}^i(L_2)+O({\Delta_{u\ra v}^i(L_2))}^2}}%
				{\sum_{j=1}^3\exp\bc{-\frac32\sum_{u\in N(v)\setminus w}\Delta_{u\ra v}^j(L_2)+
					O({\Delta_{u\ra v}^j(L_2))}^2}}\nonumber\\
			&=&\brk{\sum_{j=1}^3\exp\brk{\frac32\sum_{u\in N(v)\setminus w}\Delta_{u\ra v}^i(L_2)-\Delta_{u\ra v}^j(L_2)+
					O(\eps d)^{-2}}}^{-1}
			\label{eq_Endgame1}
	\end{eqnarray}
Since for any $v$ we have $|N(v)|=2d$, we can essentially neglect the $O(\eps d)^{-2}$-term in~(\ref{eq_Endgame1}).
More precisely, for some $-\eps^2\leq\gamma_2=\gamma_{2}(i,v,w)\leq\eps^2$ we have
	\begin{equation}\label{eq_Endgame2}
	\eta_{v\ra w}^i(L_2+1)=(1+\gamma_{2})\brk{\sum_{j=1}^3\exp\brk{\frac32\sum_{u\in N(v)\setminus w}\Delta_{u\ra v}^i(L_2)-
		\Delta_{u\ra v}^j(L_2)}}^{-1}.
	\end{equation}

To analyze~(\ref{eq_Endgame2}), assume without loss of generality that $v\in V_1$.
Then
(\ref{eq_EndgameI}) and~(\ref{eq_EndgameII}) entail that there is a number $-\eps^2<\gamma_3<\eps^2$ such that
	\begin{eqnarray*}
	\sum_{u\in N(v)\setminus w}\Delta_{u\ra v}^1(L_2)-\Delta_{u\ra v}^2(L_2)&=&
		-\bc{\frac32+\gamma_3}d\mu.
	\end{eqnarray*}	
Consequently, $\eta_{v\ra w}^1(L_2+1)=(1+\gamma_{2})\brk{1+2\exp\bc{-(3/2+\gamma_3)d\mu}}^{-1}$.
Finally, since $\mu\leq2(\eps d)^{-1}$, we obtain
	\begin{equation}\label{eq_Endgame3}
	\eta_{v\ra w}^1(L_2+1)=(1+\gamma_{4})\brk{1+2\exp\bc{-\frac32d\mu}}^{-1}
	\end{equation}
for some $-2\eps^2\leq\gamma_4=\gamma_4(1,v,w)\leq2\eps^2$.

Now, assume that $v\in V_2$.
Then
(\ref{eq_EndgameI}) and~(\ref{eq_EndgameII}) entail that there are numbers $-\eps^2<\gamma_5,\gamma_6<\eps^2$ such that
	\begin{eqnarray*}
	\sum_{u\in N(v)\setminus w}\Delta_{u\ra v}^1(L_2)-\Delta_{u\ra v}^3(L_2)&=&\gamma_5d\mu,\\
	\sum_{u\in N(v)\setminus w}\Delta_{u\ra v}^1(L_2)-\Delta_{u\ra v}^2(L_2)&=&(3/2+\gamma_6)d\mu.
	\end{eqnarray*}
Therefore,
	\begin{equation}\label{eq_Endgame4}
	\eta_{v\ra w}^2(L_2+1)=(1+\gamma_4)\brk{2+\exp\bc{\frac32d\mu}}^{-1}
	\end{equation}
for some $-2\eps^2\leq\gamma_4=\gamma_4(2,v,w)\leq2\eps^2$.
Combining~(\ref{eq_Endgame3}) and~(\ref{eq_Endgame4}), we obtain the assertion.
\end{proof}

\begin{cor}\label{Cor_EndgameMedium}
Suppose that. $0.01\eps d^{-1}\leq\|\Delta(L_2)\|_\infty\leq(\eps d)^{-1}$.
Then $\eta_{v\ra w}^a(L_2+2)\geq0.99$ if $v\in V_a$, and $\eta_{v\ra w}^a(L_2+2)\leq0.01$ if $v\not\in V_a$.
\end{cor}
\begin{proof}
We assume without loss of generality that $a=1$.
Moreover, suppose that $v\in V_1$.
We shall bound the quotient
	\begin{eqnarray}
	\frac{\eta_{v\ra w}^1(L_2+2)}{\eta_{v\ra w}^2(L_2+2)}&=&Q_2\cdot Q_3,\mbox{ where}\\
	Q_j&=&\prod_{u\in V_j\cap N(v)\setminus w}\frac{1-\eta_{u\ra v}^1(L_2+1)}{1-\eta_{u\ra v}^2(L_2+1)}\mbox{ for }j=2,3,\nonumber
	\end{eqnarray}
from below.
Lemma~\ref{Lemma_EndgameMedium} implies that for $u\in V_3$
	\begin{eqnarray*}
	\frac{1-\eta_{u\ra v}^1(L_2+1)}{1-\eta_{u\ra v}^2(L_2+1)}&\geq&
		\frac{2/3+(1-\eps^3)\beta'}{2/3+(1+\eps^3)\beta'}\geq1+3\eps^3\beta'\geq1-6\eps^3.
	\end{eqnarray*}
Hence,
	\begin{equation}\label{eqCorEndgameMedium1}
	Q_2\geq(1-6\eps^3)^d.
	\end{equation}
Furthermore, for $u\in V_2$ Lemma~\ref{Lemma_EndgameMedium} entails that
	\begin{eqnarray*}
	\frac{1-\eta_{u\ra v}^1(L_2+1)}{1-\eta_{u\ra v}^2(L_2+1)}&\geq&
		\frac{2/3+(1-\eps^3)\beta'}{2/3-(1+\eps^3)\beta}=1+\frac{(1-\eps^3)(\beta+\beta')}{2/3-(1-\eps^3)\beta}\geq1+2\eps^2.
	\end{eqnarray*}
Consequently,
	\begin{equation}\label{eqCorEndgameMedium2}
	Q_2\geq(1+2\eps^2)^{d-1}.
	\end{equation}
Combining~(\ref{eqCorEndgameMedium1}) and~(\ref{eqCorEndgameMedium2}) and recalling that $d\gg\eps^{-2}$, we obtain the assertion.
\end{proof}

\begin{lem}\label{Lemma_EndgameLarge}
Suppose that $\|\Delta(L_2)\|_\infty>(\eps d)^{-1}$.
Then $\eta_{v\ra w}^a(L_2+1)\geq0.99$ if $v\in V_a$, and $\eta_{v\ra w}^a(L_2+2)\leq0.01$ if $v\not\in V_a$.
\end{lem}
\begin{proof}
Since $\|\Delta(L_2)\|_\infty>(\eps d)^{-1}$, (\ref{eq_EndgameI}) and~(\ref{eq_EndgameII}) yield
	\begin{equation}\label{eqEndgameLarge}
	\mu\geq(2\eps d)^{-1}.
	\end{equation}
Without loss of generality we may consider a vertex $v\in V_1$ and a neighbor $w\in N(v)$.
We will prove that $\eta_{v\ra w}^1(L_2+1)/\eta_{v\ra w}^2(L_2+1)>1000$.
Since $\sum_{j=1}^3\eta_{v\ra w}^j(L_2+1)=1$, this implies the assertion.
To bound the quotient from below, we decompose
	\begin{eqnarray}
	\frac{\eta_{v\ra w}^1(L_2+1)}{\eta_{v\ra w}^2(L_2+1)}&=&Q_2\cdot Q_3,\mbox{ where}\\
	Q_j&=&\prod_{u\in V_j\cap N(v)\setminus w}\frac{1-\eta_{u\ra v}^1(L_2)}{1-\eta_{u\ra v}^2(L_2)}\mbox{ for }j=2,3,\nonumber
	\end{eqnarray}
With respect to $Q_3$, (\ref{eq_EndgameI}) and~(\ref{eq_EndgameII}) imply that for $u\in V_3$
	\begin{eqnarray*}
	\frac{1-\eta_{u\ra v}^1(L_2)}{1-\eta_{u\ra v}^2(L_2)}&\geq&\frac{2/3+(1/2-\eps^3)\mu)}{2/3+(1/2+\eps^3)\mu}
		=1-\frac{2\eps^3\mu}{2/3+(1/2+\eps^3)\mu}\geq1-3\eps^3\mu.
	\end{eqnarray*}
Hence,
	\begin{equation}\label{eqEndgameLarge1}
	Q_3\geq(1-3\eps^3\mu)^d.
	\end{equation}
Further, (\ref{eq_EndgameI}) and~(\ref{eq_EndgameII}) yield that for $u\in V_2$
	\begin{eqnarray*}
	\frac{1-\eta_{u\ra v}^1(L_2)}{1-\eta_{u\ra v}^2(L_2)}&\geq&\frac{2/3+(1/2-\eps^3)\mu)}{2/3-(1-\eps^3)\mu}
		=1+\frac{(3/2-2\eps^3)\mu}{2/3+(1/2-(1-\eps^3))\mu}\geq1+2\mu.
	\end{eqnarray*}
Therefore,
	\begin{equation}\label{eqEndgameLarge2}
	Q_2\geq(1+2\mu)^{d-1}.
	\end{equation}
Thus, combining~(\ref{eqEndgameLarge})--(\ref{eqEndgameLarge2}), we obtain
	\begin{eqnarray*}
	\frac{\eta_{v\ra w}^1(L_2+1)}{\eta_{v\ra w}^2(L_2+1)}&=&Q_2\cdot Q_3
		\geq(1-3\eps^3\mu)^d(1+2\mu)^{d-1}\geq(1+\mu)^{d-1}\geq1000,
	\end{eqnarray*}
which implies the assertion.
\end{proof}
\noindent
Finally, Proposition~\ref{Prop_Endgame} is a direct consequence of
Corollary~\ref{Cor_EndgameMedium} and Lemma~\ref{Lemma_EndgameLarge}.

\section{Proof of Corollary~\ref{cor:BPforCol}}\label{sec:cor_BPforCol}

\emph{Throughout this section, we assume that $d\geq d_0$ for a sufficiently large constant $d_0>0$,
and that $n>n_0=n_0(d)$ for a large enough $n_0$.
Set $p=d/n$.}


Let $G=G_{n,d,3}$ be a random graph with vertex set $V=\{1,\ldots,3n\}$
and ``planted'' 3-coloring $V_1,V_2,V_3$.
In order to analyze the adjacency $A(G)$, we shall employ the following lemma,
which follows immediately from the ``converse expander mixing lemma'' from~\cite{Bilu}.

\begin{lem}\label{Thm_Bilu}
Let $B=(V'\du V'',E_B)$ be a bipartite $d$-regular graph such that $|V'|=|V''|$.
Assume that
	\begin{equation}\label{eqBilu}
	\forall S\subset V',\,T\subset V'':|e_B(S,T)-|S||T|p|\leq d^{0.51}\sqrt{|S||T|},
	\end{equation}
where $e_B(S,T)$ is the number of $S$-$T$-edges in $B$.
Then the adjacency matrix $A(B)$ enjoys the property:
\begin{quote}
for any two vectors $\xi,\eta\in\RR^{V'\cup V''}$ such that both $\xi$, $\eta$ are perpendicular
to $\vecone_{V'}$ and $\vecone_{V''}$ the inequality
	$\scal{A(B)\xi}{\eta}\leq d^{0.52}\|\xi\|\|\eta\|$
holds.
\end{quote}
\end{lem}
\noindent
Moreover, 
the following lemma can be derived using standard techniques from the
theory of random regular graphs~\cite[Chapter~9]{JLR}.

\begin{lem}\label{Lemma_disc}
W.h.p.\ $G$ has the following property.
Let $1\leq i<j\leq 3$.
Then
	$$\forall S\subset V_i,\,T\subset V_j:|e_G(S,T)-|S||T|p|\leq d^{0.51}\sqrt{|S||T|}.$$
\end{lem}

\begin{cor}\label{Cor_reg}
W.h.p.\ $G$ is $(d,0.01)$-regular.
\end{cor}
\begin{proof}
Let $A(G)=(a_{v,w})_{v,w\in V}$ denote the adjacency matrix of $G$.
Moreover, let
	$$a_{vw}^{ij}=\left\{\begin{array}{cl}
		a_{vw}&\mbox{ if }v,w\in V_i\cup V_j,\\
		0&\mbox{ otherwise}
		\end{array}\right.\qquad(1\leq i<j\leq 3).$$
Then $A^{ij}=(a_{vw}^{ij})_{v,w\in V}$ is the adjacency matrix of the bipartite subgraph of $G$
induced on $V_i\cup V_j$.
Let $\EE$ be the subspace of $\RR^{V}$ spanned by the three vectors $\vecone_{V_k}$ ($k=1,2,3$).
Combining Lemma~\ref{Thm_Bilu} with Lemma~\ref{Lemma_disc}, we conclude that w.h.p.\
	$\scal{A^{ij}\xi}{\eta}\leq d^{0.52}\|\xi\|\|\eta\|$ for all $\xi,\eta\perp\EE$ and any $1\leq i<j\leq 3$.
Since $A(G)=\sum_{1\leq i<j\leq 3}A^{ij}$, this implies that
	\begin{equation}\label{eqPrfReg1}
	\forall \xi,\eta\perp\EE:\scal{A(G)\xi}{\eta}\leq 0.01d\|\xi\|\|\eta\|
	\end{equation}
(provided that $d$ is sufficiently large).
Furthermore, as the construction of $G$ ensures that each vertex $v\in V_i$ has exactly
$d$ neighbors in each class $V_j\not=V_i$, we can compute the vector $\zeta^i=A(G)\vecone_{V_i}$ as follows:
for any $v\in V$
	\begin{eqnarray*}
	\zeta_v^i&=&\sum_{w\in V_i}a_{vw}=\left\{\begin{array}{cl}0&\mbox{ ifÊ}v\in V_i,\\d&\mbox{ if }v\not\in V_i.
			\end{array}\right.
	\end{eqnarray*}
Hence, $\zeta^i=A(G)\vecone_{V_i}=d\sum_{j\not=i}\vecone_{V_j}$.
Therefore, for any $1\leq i<j\leq3$ we have
	\begin{equation}\label{eqPrfReg2}
	A(G)(\vecone_{V_i}-\vecone_{V_j})=-d(\vecone_{V_i}-\vecone_{V_j}).
	\end{equation}
Combining~(\ref{eqPrfReg1}) and~(\ref{eqPrfReg2}), we see that $G$ is $(d,0.01)$-regular w.h.p.
\end{proof}
\noindent
Finally, Corollary~\ref{cor:BPforCol} follows from Theorem~\ref{thm:BPforCol} and Corollary~\ref{Cor_reg}.


\section{Conclusion}

We have shown that \texttt{BPCol} 3-colors $(d,0.01)$-regular graphs in polynomial time.
Three potentially interesting extensions suggest themselves, which may be the subject of future work.
\begin{enumerate}
\item In $(d,0.01)$-regular graphs every vertex has precisely $d$ neighbors in each color class
	except for its own.
	By comparison, in the planted random graph model studied in~\cite{AlonKahale97} the number
	that a vertex has in another color class is Poisson with mean $d$.
	It would be interesting to see if/how the present analysis can be modified to deal with such a
	more irregular degree distribution.
\item Survey Propagation (``SP'') is a more involved version of Belief Propagation (although SP can
	be rephrased as BP on a different model~\cite{ANewLook}) and performs very well empirically
	on random graphs $G(n,p)$.
	It would be interesting to extend our analysis to SP.
\item In a $(d,0.01)$-regular graph there is exactly one 3-coloring (up to permutations of the color classes).
	Nonetheless, we think that the techniques of our analysis can be extended to more complicated
	``solution spaces''.
	For instance, it should be straightforward to deal with graphs that have a bounded number of distinct 3-colorings.
\end{enumerate}


\end{document}